\documentclass[10pt,twocolumn,twoside]{IEEEtran}
\usepackage{cite}
\usepackage{amsmath,amssymb,amsfonts}
\usepackage{algorithmic}
\usepackage{graphicx}
\usepackage{textcomp}
\usepackage{mathtools}
\usepackage{bm}
\usepackage[percent]{overpic}
\usepackage{amsthm}
\usepackage{xcolor}
\usepackage{eso-pic}

\definecolor{IEEEtcnsblue}{RGB}{0,0,0}

\newtheoremstyle{tcnsplain}
{0pt}
{0pt}
{\itshape}
{}
{\sffamily\bfseries\itshape\color{IEEEtcnsblue}\fontsize{8.5}{9.5}\selectfont}
{:}
{0.5em}
{\thmname{#1}~\thmnumber{#2}\thmnote{~(#3)}}

\newtheoremstyle{tcnsdefinition}
{0pt}
{0pt}
{\normalfont}
{}
{\sffamily\bfseries\itshape\color{IEEEtcnsblue}\fontsize{8.5}{9.5}\selectfont}
{:}
{0.5em}
{\thmname{#1}~\thmnumber{#2}\thmnote{~(#3)}}

\theoremstyle{tcnsplain}
\newtheorem{theorem}{Theorem}
\newtheorem{lemma}{Lemma}
\newtheorem{proposition}{Proposition}
\newtheorem{corollary}{Corollary}

\theoremstyle{tcnsdefinition}
\newtheorem{definition}{Definition}
\newtheorem{remark}{Remark}

\renewenvironment{proof}[1][Proof]{%
  \par\noindent{\sffamily\bfseries\itshape\color{IEEEtcnsblue}\fontsize{8.5}{9.5}\selectfont #1:}\hspace{0.5em}\normalfont
}{\hfill$\blacksquare$\par}

\def\BibTeX{{\rm B\kern-.05em{\sc i\kern-.025em b}\kern-.08em
    T\kern-.1667em\lower.7ex\hbox{E}\kern-.125emX}}

\usepackage[colorlinks=true,
            linkcolor=blue,
            citecolor=blue,
            urlcolor=blue]{hyperref}

\hypersetup{
    colorlinks=true,
    linkcolor=blue,
    citecolor=blue,
    urlcolor=blue
}

\begin{document}

\AddToShipoutPictureFG*{%
  \AtPageLowerLeft{%
    \raisebox{8mm}{%
      \makebox[\paperwidth][c]{%
        \parbox{\textwidth}{%
          \centering\small
          This work has been submitted to the IEEE for possible publication.
          Copyright may be transferred without notice, after which this version
          may no longer be accessible.
        }%
      }%
    }%
  }%
}

\title{Analysis and Consensus Control of Emergent Dynamic Polarization in Minimally-Nonlinear Opinion Dynamics}

\author{Rajul Kumar, \IEEEmembership{Graduate Student Member, IEEE}, and Ningshi Yao, \IEEEmembership{Member, IEEE}
\thanks{This work was partially supported by NIH 1R01EB038710-01, ONR N00014-26-1-2153, and NSF CAREER 2539218.}
\thanks{The authors are with the Department of Electrical and Computer Engineering, George Mason University, Fairfax, VA 22030 USA (e-mail: \{rkumar29, nyao4\}@gmu.edu). Corresponding author: Ningshi Yao.}}

\maketitle
\thispagestyle{empty}

\begin{abstract}
Collective opinions in social networks evolve through local interaction rules, yet how such local updates give rise to dynamic polarization—persistent oscillatory disagreement between opposing opinion clusters at the network level—remains unexplained. This paper proposes a Minimally-Nonlinear Opinion Dynamics (or M-NOD) framework that analytically characterizes dynamic polarization as a truly emergent collective behavior arising solely from local, agent-level opinion-update rules without externally imposed mechanisms. By introducing a minimal cubic nonlinearity, we rigorously prove that, as the reactivity rate exceeds a critical threshold, the network’s consensus equilibrium loses stability via a supercritical flip bifurcation. In the post-bifurcation regime, this instability gives rise to a unique, locally asymptotically stable periodic orbit, thereby characterizing symmetric dynamic polarization with balanced bipartite opinion clusters. We further establish the structural robustness of this behavior by proving the existence and local asymptotic stability of asymmetric dynamic polarization under directed graphs with nonuniform influence weights. Finally, to resolve this undesirable cyclic deadlock, we develop local agent-level control strategies. We prove that anchoring the opinion of only a single agent is sufficient to eliminate network-wide oscillatory disagreement and restore asymptotically stable consensus. Numerical simulations substantiate the theoretical analysis, including the emergence of symmetric and asymmetric dynamic polarization, and demonstrate the efficacy of the proposed control interventions.
\end{abstract}

\begin{IEEEkeywords}
Opinion Dynamics, Polarization, Consensus, Networked dynamical systems
\end{IEEEkeywords}


\section{Introduction}
\label{sec:introduction}
\IEEEPARstart{T}{he} analysis of collective behavior in social networks using opinion dynamics has received increasing attention \cite{b1a, b1b, b1c}. In these settings, bounded-rational agents update their opinions based on local interactions, and collective behavior emerges from agent-level opinion updates \cite{b1d, b1e}. Polarization is a fundamental emergent behavior, in which agents' opinions diverge into opposing clusters that asymptotically converge to distinct equilibria \cite{b5, b1f, b1g}. Beyond such static polarized states, an important question is whether disagreement can remain persistently oscillatory rather than converge to a static state. In this research, we characterize this behavior as dynamic polarization: an emergent asymptotic phenomenon in which two opposing opinion clusters exhibit bounded, synchronized oscillations, such that the shift of one cluster induces a corresponding counter-shift in the other. Understanding how dynamic polarization emerges purely from local interaction rules is essential to analyze the stability of collective opinion dynamics. As dynamic polarization represents an undesirable state of persistent oscillatory disagreement and cyclic deadlock, characterizing its emergence and stability is a critical prerequisite for designing control laws that guarantee the network's transition back to consensus. Although prior literature primarily analyzes convergence to consensus, it generally does not characterize the asymptotic behavior that emerges when parameter variations cause the consensus equilibrium to lose stability. Consequently, few studies in opinion dynamics develop a rigorous characterization of how bounded, persistent oscillations can emerge as a polarized network state.

Classical weighted-averaging models constitute one of the most widely studied frameworks in opinion dynamics due to their analytical tractability, yet they typically predict convergence to consensus, static polarization, or fragmentation~\cite{b1}. To model polarization, more recent studies rely on mechanisms such as signed graphs with negative edges~\cite{b2, b3, b4}, asymmetric influence and heterogeneous susceptibility~\cite{b5}, confirmation bias~\cite{b6, b7}, game-theoretic utilities~\cite{b8}, and epidemic-type dynamics~\cite{b9}. These frameworks can only model static asymptotic states and thus cannot characterize dynamic, oscillatory polarization. Furthermore, polarization in these models is not truly emergent, as it does not arise solely from the underlying opinion-update dynamics; rather, it depends on additional mechanisms imposed a priori, such as negative interactions in signed graphs, stubborn agents, confirmation bias, or game-theoretic payoff structures.

While the aforementioned studies are restricted to static polarization, a separate class of recent frameworks has been developed to capture oscillatory opinion states and nonstationary polarization. However, these behaviors are typically sustained by external or structurally imposed mechanisms rather than emerging autonomously from purely local, agent-level opinion updates. For instance, stochastic gossip models can generate opinion fluctuations by introducing fixed stubborn agents~\cite{b10}, but the resulting behavior is stochastic disagreement rather than oscillatory polarization. External processes, such as delayed information and cognitive bias, can model fluctuating opinions without capturing polarization~\cite{b11, b12}. Similarly, self-sustained oscillations have been observed in hypergraph-based opinion models~\cite{b13}, where the entire network evolves in a coherent collective pattern. Memory-based bounded-confidence models can produce oscillatory collective dynamics with polarized groups~\cite{b14}, although these oscillations rely on state-dependent network switching and extreme-agent influence. Models in~\cite{b15, b16} predict oscillatory dynamics within polarized populations; however, these arise from external environmental feedback loops and auxiliary systems that function as engineered constructs. The closely related work in~\cite{b17} studies dynamic polarization, referring to the phenomenon as nonstationary polarization, but the mechanism relies on explicit structural hostility, requiring negative edges that represent competition and partition the network into two opposing camps. Similarly, while nonlinear feedback models have captured fluctuating asymmetric polarization via bifurcations~\cite{b17a}, these transitions remain driven by external inputs, such as shifts in public policy mood. Consequently, to the best of the authors' knowledge, the current literature lacks an analytical framework that explains the emergence of dynamic polarization from purely local, agent-level opinion-update rules. This leaves open the fundamental question of whether dynamic polarization can arise as a truly emergent phenomenon without relying on externally imposed mechanisms, motivating the need for a new analytical opinion dynamics framework capable of explaining this behavior.

Discovering such an opinion dynamics rule capable of generating dynamic polarization as a collective, network-wide emergent behavior is challenging. Such dynamics must maintain bounded opinions, preserve local interaction structure, explain how consensus loses stability, and still allow analytical characterization of nonstationary periodic attractors. Arbitrary nonlinear rules may lead to unboundedness, chaos, or analytical intractability. Consensus control is equally challenging because the network is trapped in a self-sustained periodic attractor, requiring localized agent-level interventions that destabilize the stable cyclic deadlock while guaranteeing convergence to asymptotically stable consensus.

Addressing this research gap, as the main contribution of this paper, we propose a novel minimally-nonlinear opinion dynamics (M-NOD) framework that characterizes dynamic polarization as a truly emergent collective behavior arising from local, agent-level opinion-update rules. Without relying on externally imposed mechanisms such as signed edges, we introduce a cubic nonlinearity as a minimal nonlinear representation in each individual agent's update rule, which in turn is capable of capturing the emergence of dynamic polarization. Within this framework, we rigorously establish that under a mean-field baseline, the M-NOD framework undergoes a flip bifurcation through which the consensus equilibrium loses stability; in a high-reactivity regime, this gives rise to a period-2 orbit where the network exhibits balanced bipartite clustering, characterizing the resulting oscillatory behavior as symmetric dynamic polarization. Additionally, by establishing the structural robustness of symmetric dynamic polarization under small symmetry-breaking perturbations, we prove the existence and stability of asymmetric dynamic polarization, where two opposing opinion groups with unequal cluster sizes remain in a persistent cyclic deadlock. 

Finally, leveraging the M-NOD framework's ability to characterize dynamic polarization provides a principled basis for the transition to consensus. To prevent the network from remaining trapped in dynamic polarization, an undesirable state of sustained disagreement and cyclic deadlock, we develop consensus control strategies that progressively reduce the required intervention: first controlling one polarized cluster, then one agent from each polarized cluster, and finally controlling only a single agent in the entire network. The consensus control results show that localized opinion control can break dynamic polarization, causing the remaining agents to align with the reference over time and driving the entire network to a locally asymptotically stable consensus.

\vspace{1mm}
\textbf{Notations:} $\mathbb{R}$, $\mathbb{R}_{\ge0}$ and $\mathbb{N}$ denote the sets of real, nonnegative real and natural numbers, respectively. Bold lowercase (e.g., $\mathbf{v}$) and uppercase letters (e.g., $\mathbf{C}$) denote vectors and matrices, respectively. $\mathbf{I}$ and $\mathbf{J}$ represent the identity matrix and all-ones matrix of appropriate dimensions. $\mathbf{1}_N$ denotes the $N$-dimensional all-ones column vector. For $\mathbf{v}\!\in\!\mathbb{R}^N$, $\|\mathbf{v}\|$ and $\|\mathbf{v}\|_\infty$ denote the Euclidean ($\ell_2$) and infinity norms, respectively. The notation $(\mathbf{v})^{\circ 3}$ denotes the element-wise (Hadamard) cube, i.e., $[(\mathbf{v})^{\circ 3}]_i\!=\!v_i^3$. Equilibria are denoted by an asterisk (e.g., $y^*$). For a differentiable map $\mathbf{G}$, $D\mathbf{G}(\cdot)$ denotes its Jacobian.


\section{Problem Setup: Minimally-Nonlinear Opinion Dynamics (M-NOD)}
Consider a network of $N\!\in\! 2\mathbb{N}$ interacting agents over a weighted graph $\mathcal{G}\!=\!(\mathcal{V},\mathcal{E},\mathbf{W})$, where $\mathcal{V}\!=\!\{1,\dots,N\}$ and $\mathcal{E}\subseteq\mathcal{V}\times\mathcal{V}$ denote the agent and edge sets, respectively. The matrix $\mathbf{W}=[w_{ij}]\in\mathbb{R}_{\ge0}^{N\times N}$ denotes the weighted adjacency matrix, where $w_{ij}>0$ indicates that agent $j$ influences agent $i$ (i.e., $(j,i) \in \mathcal{E}$), and $w_{ij}\!=\!0$ otherwise. This formulation accommodates directed communication topologies, reducing to an undirected graph if and only if $\mathbf{W}$ is symmetric. Each agent $i\! \in\! \mathcal{V}$ holds an opinion $v_i(t)\!\in\!\mathbb{R}$, and the collective network opinion state is $\mathbf{v}_t\!=\![v_1(t),\dots,v_N(t)]\!^\top\!\in\!\mathbb{R}^N$. The minimally-nonlinear opinion dynamics (M-NOD) evolve as,
\begin{equation}\label{eq:1}
    \mathbf{v}_{t+1} = [(1-\gamma)\mathbf{I} + \gamma \mathbf{W}]\mathbf{v}_{t} + \gamma a\,(\mathbf{C}\mathbf{v}_{t})^{\circ 3},
\end{equation}
where $\gamma, a\!\in\! \mathbb{R}_{\ge0}$ denote the reactivity rate and opinion nonlinearity strength, respectively. The matrix $\mathbf{C}\!:=\!\mathbf{I}-\mathbf{W}$ acts as the interaction-induced deviation operator, where $(\mathbf{C}\mathbf{v}_t)_i\!=\!v_i(t)-\sum_{j=1}^N w_{ij}v_j(t)$ is the deviation of agent $i$'s opinion from the weighted aggregate opinion of its neighbors $\mathcal{N}_i\!:=\!\{j\in\mathcal{V}\!:\!(j,i)\in\mathcal{E}\}$. For row-stochastic $\mathbf{W}$, $\mathbf{C}$ is the normalized graph Laplacian.

M-NOD~\eqref{eq:1} models opinion evolution in a network of bounded-rational agents through linear social influence and a nonlinear deviation-driven correction. The term $((1\!-\!\gamma)\mathbf{I}\!+\!\gamma\mathbf{W})\mathbf{v}_t$, equivalent to $\mathbf{v}_t \!-\! \gamma\mathbf{C}\mathbf{v}_t$, represents the linear social-influence update. Here, each agent adjusts its opinion proportionally to the deviation between its current opinion and the weighted aggregate opinion of its neighbors, with $\gamma$ controlling the reactivity to this assimilation. The component $\gamma a(\mathbf{C}\mathbf{v}_t)^{\circ 3}$ introduces a minimally nonlinear cubic correction driven by the local interaction-induced deviation $\mathbf{C}\mathbf{v}_t$. It produces a saturation effect that moderates the linear social influence under large disagreements, thereby introducing a nonlinear damping that prevents excessively large opinion adjustments that could otherwise induce divergent behavior. Larger deviations from the weighted neighbor aggregate produce stronger nonlinear resistance. In the absence of nonlinearity ($a\!=\!0$), M-NOD \eqref{eq:1} reduces to the linear Rescorla--Wagner (RW) multi-agent associative learning, which follows directly from extending the single-agent RW rule in \cite{b18}. Appendix~\ref{app1} derives \eqref{eq:1} from the corresponding single-agent opinion update.

The selection of the cubic nonlinearity is mathematically motivated as the simplest minimal nonlinear representation capable of capturing the emergence of dynamic polarization in opinion dynamics. While purely linear error-correction models unphysically predict unbounded opinion divergence under high reactivity, any realistic update rule must incorporate a saturating mechanism to bound opinions. Rather than relying on complex, domain-specific transcendental functions (e.g., sigmoids or hyperbolic tangents), the cubic term serves as the canonical normal-form representation of a broad class of smooth, odd-symmetric saturating nonlinearities near the bifurcation point. By acting as the leading-order Taylor expansion of these bounded response functions, the cubic term isolates the minimal mathematical mechanism necessary to trigger and sustain dynamic polarization, thereby preserving both analytical tractability and theoretical generality.

To characterize the asymptotic behavior of the network under M-NOD, we formally define the emergent states as:

\vspace{1mm}

\begin{definition}\label{def:2.3}
The network is said to achieve consensus if the collective opinion state $\mathbf{v}_t$ converges asymptotically to a uniform steady state, i.e., there exists $v_c\in\mathbb{R}$ such that $\lim_{t\to\infty}\mathbf{v}_t=v_c\mathbf{1}_N$. In the uncontrolled average-preserving case, $v_c=\bar v_0$, where $\bar v_0=\frac{1}{N}\mathbf{1}_N^\top\mathbf{v}_0$. Equivalently, consensus implies that the interaction-induced deviation vanishes asymptotically, such that $\lim_{t\to\infty}\mathbf{C}\mathbf{v}_t=\mathbf{0}$.
\end{definition}

\vspace{1mm}

\begin{definition}\label{def:2.1}
The network is said to exhibit symmetric dynamic polarization if its opinion trajectory $\mathbf{v}_t$ satisfies:\\
1. Converges to a stable period-2 orbit $\mathcal{O}\coloneqq\{\bar{\mathbf{v}},-\bar{\mathbf{v}}\}$, with $\mathbf{v}_t=\bar{\mathbf{v}}$ and $\mathbf{v}_{t+1}=-\bar{\mathbf{v}}$, returning to $\mathbf{v}_{t+2}=\bar{\mathbf{v}}$, thus producing the asymptotic oscillation $\mathbf{v}_{t+2}=\mathbf{v}_t\neq\mathbf{v}_{t+1}$.\\
2. The period-2 attractor lies entirely in the zero-mean disagreement subspace, i.e., $\mathbf{1}_N^\top\bar{\mathbf{v}}=0$, with $\bar{\mathbf{v}}\neq\mathbf{0}$.\\
The network remains partitioned into two opposing opinion groups with equal magnitudes and persistent oscillatory disagreement, exhibiting antisymmetry.
\end{definition}

\vspace{1mm}

\begin{definition}\label{def:2.2}
The network is said to exhibit asymmetric dynamic polarization if its opinion trajectory $\mathbf{v}_t$ satisfies:\\
1. Converges to a stable period-2 orbit $\mathcal{O}_\epsilon \coloneqq \{\bar{\mathbf{v}}_\epsilon, \tilde{\mathbf{v}}_\epsilon\}$, with $\mathbf{v}_t = \bar{\mathbf{v}}_\epsilon$ and $\mathbf{v}_{t+1} = \tilde{\mathbf{v}}_\epsilon$, returning to $\mathbf{v}_{t+2} = \bar{\mathbf{v}}_\epsilon$, thereby yielding the asymptotic oscillation $\mathbf{v}_{t+2} = \mathbf{v}_t \neq \mathbf{v}_{t+1}$.\\
2. The period-2 attractor does not lie entirely in the zero-mean disagreement subspace, i.e., at least one orbit state satisfies $\mathbf{1}_N^\top\bar{\mathbf{v}}_\epsilon \neq 0$ or $\mathbf{1}_N^\top\tilde{\mathbf{v}}_\epsilon \neq 0$, with $\bar{\mathbf{v}}_\epsilon,\tilde{\mathbf{v}}_\epsilon \neq \mathbf{0}$.

The network remains partitioned into two opposing opinion groups with unequal opinion amplitudes and/or cluster sizes.
\end{definition}

\vspace{1mm}

\textbf{Problem Statement:} The objective is twofold: first, to analytically characterize the emergence and stability of symmetric and asymmetric dynamic polarization in M-NOD \eqref{eq:1}; and second, to design localized agent-level control laws that guarantee the network's transition from dynamic polarization to asymptotically stable consensus.

Having formalized the problem, Section~\ref{sec3} establishes the exact parametric conditions, specifically the critical reactivity and bifurcation thresholds, under which symmetric dynamic polarization emerges.


\section{Main Results I: Emergence of Symmetric Dynamic Polarization over Undirected Complete Graphs with Uniform Weights}\label{sec3}

This section analyzes M-NOD over an undirected complete graph with uniform influence weights, representing the maximally connected mean-field baseline. This idealized setting provides the simplest analytical setting for demonstrating that dynamic polarization arises as an emergent behavior purely from the minimal nonlinearity, independently of structural graph disconnection, sign-based negative ties, or weight heterogeneity. Specifically, we consider the edge set $\mathcal{E}=\{(j,i)\mid j\neq i\}$ and uniform weights $w_{ij}=\frac{1}{N-1}$ for $j\neq i$, with $w_{ii}=0$. Hence,
$\mathbf{W}=\frac{1}{N-1}(\mathbf{J}-\mathbf{I})$, which is doubly stochastic, i.e., $\mathbf{W}\mathbf{1}_N=\mathbf{1}_N$ and $\mathbf{1}_N^\top\mathbf{W}=\mathbf{1}_N^\top$.

This mean-field baseline serves as the analytical foundation for evaluating the robustness of symmetric dynamic polarization to directed graph topologies and nonuniform influence weights, as explored in Section \ref{sec4}.

Lemma~\ref{lem:3.1} considers a necessary subspace $\mathcal{S}$ that enforces the zero-sum property of the collective network disagreement. This condition ensures that no unaccounted disagreement is introduced by any agent; instead, all disagreement remains internal to the network, with each agent's disagreement arising solely from its interactions with other agents within the network. Crucially, the lemma establishes that any uniform disagreement state, where every agent has the same disagreement magnitude $|(\mathbf{C}\mathbf{v})_i|=p>0$, must lie in this subspace $\mathcal{S}$. Since the total disagreement must be balanced and sum to zero, this uniformity partitions the network into two equal-sized, balanced bipartite clusters.

\vspace{1mm}

\begin{lemma}\label{lem:3.1}
Under M-NOD \eqref{eq:1} over an undirected complete graph with uniform influence weights, any non-consensus uniform disagreement state $\mathbf{v} \notin \text{span}\{\mathbf{1}_N\}$ satisfies $\mathbf{v} \in \mathcal{S}$, where the symmetric subspace $\mathcal{S}$ is defined as:
\begin{equation}\label{eq:symmetric_subspace}
    \mathcal{S} \coloneqq \left\{\,\mathbf{v}\!\in\!\mathbb{R}^N \;\middle|\;
    \mathbf{1}_N^\top\!(\mathbf{C}\mathbf{v})\!=\!0 \;\text{and}\;
    \mathbf{1}_N^\top\!\bigl((\mathbf{C}\mathbf{v})^{\circ 3}\bigr)\!=\!0 \right\}.
\end{equation}
Consequently, such a state exhibits balanced bipartite clustering, partitioning the network into two equal-sized opposing opinion clusters.
\end{lemma}

\vspace{1mm}

\begin{proof}
Let $\mathbf{y_t}\!\coloneqq\!\mathbf{C}\mathbf{v_t}$ denote the disagreement vector associated with a non-consensus uniform disagreement state $\mathbf{v}\!\notin\!\mathrm{span}\{\mathbf{1}_N\}$. We verify that $\mathbf{y}_t\!\coloneqq\!\mathbf{C}\mathbf{v}_t$ satisfies the orthogonality conditions defining the symmetric subspace~$\mathcal{S}$. Substituting $\mathbf{C}\!=\!\mathbf{I}-\mathbf{W}$ into $\mathbf{y}_t\!=\!\mathbf{C}\mathbf{v}_t$ yields
$\mathbf{1}_N^\top\mathbf{y}_t\!=\!(\mathbf{1}_N^\top\mathbf{I}-\mathbf{1}_N^\top\mathbf{W})\mathbf{v}_t$.
Under the complete undirected graph topology with uniform influence weights, we have
$\mathbf{1}_N^\top\mathbf{y}_t\!=\!(\mathbf{1}_N^\top-\mathbf{1}_N^\top)\mathbf{v}_t\!=\!0$.

Next, the uniformity of the disagreement magnitude implies that for every
component $i$, we have $|y_i|\!=\!p$ for some constant $p>0$, and consequently
$y_i^{2}\!=\!p^{2}$. Thus, the cubic term factors as
$y_i^{3}\!=\!y_i\cdot y_i^{2}\!=\!p^{2}y_i$. In vector form, this yields the identity
$(\mathbf{y})^{\circ 3}\!=\!p^{2}\mathbf{y}$. Substituting into
$\mathbf{1}_N^\top(\mathbf{y})^{\circ 3}$ gives
$\mathbf{1}_N^\top(\mathbf{y})^{\circ 3}\!=\!p^{2}(\mathbf{1}_N^\top\mathbf{y})\!=\!p^{2}(0)\!=\!0$. Hence, $\mathbf{1}_N^\top(\mathbf{y})^{\circ 3}\!=\!0$. Finally, substituting $\mathbf{y}\!=\!\mathbf{C}\mathbf{v}$ into the obtained identities
$\mathbf{1}_N^\top(\mathbf{y})\!=\!0$ and $\mathbf{1}_N^\top(\mathbf{y})^{\circ 3}\!=\!0$ yields $\mathbf{1}_N^\top(\mathbf{C}\mathbf{v})\!=\!0$ and
$\mathbf{1}_N^\top\!((\mathbf{C}\mathbf{v})^{\circ 3})\!=\!0$. Hence, any non-consensus uniform disagreement state $\mathbf{v}$ lies within the symmetric subspace $\mathcal{S}$ in~\eqref{eq:symmetric_subspace}.

Since the components of the disagreement vector $\mathbf{y}$ are restricted to the values $\pm p$ (where $p > 0$ is the uniform magnitude), the mean-zero condition $\mathbf{1}_N^\top \mathbf{y} = 0$ implies that the sum of positive disagreements must exactly offset the sum of negative disagreements. This necessitates that the number of agents holding opinion $+p$ equals the number of agents holding $-p$. Consequently, the network partitions into two opposing clusters of equal cardinality, establishing the existence of balanced bipartite clustering.
\end{proof}

Note that for symmetric dynamic polarization, uniform disagreement is the defining structural property under which the two opposing opinion groups exhibit equal-magnitude disagreement with opposite signs. Furthermore, for uniform balanced disagreement trajectories, the subspace $\mathcal{S}$ is invariant under M-NOD~\eqref{eq:1}.

The opinion state $\mathbf{v}_t$ in~\eqref{eq:1} can be orthogonally decomposed into an average consensus component and a global deviation component. We
define the average consensus as $\bar{v}_t\! \coloneqq\!\tfrac{1}{N}\mathbf{1}_N^\top \mathbf{v}_t$ and the mean-zero deviation as $\mathbf{z}_t$, where $\mathbf{1}_N^\top\mathbf{z}_t\!=\!0$. To project the dynamics onto the consensus subspace, we multiply~\eqref{eq:1} by $\tfrac{1}{N}\mathbf{1}_N^\top$ and use the identity $\mathbf{1}_N^\top\mathbf{W}\!=\!\mathbf{1}_N^\top$. This yields:
\[
\begin{aligned}
\frac{1}{N}\mathbf{1}_N^\top \mathbf{v}_{t+1}
    &= \frac{1}{N}\mathbf{1}_N^\top\!\big[(1-\gamma)\mathbf{I}+\gamma\mathbf{W}\big]\mathbf{v}_t
      + \frac{\gamma a}{N}\mathbf{1}_N^\top (\mathbf{y}_t)^{\circ 3}, \\
\bar{v}_{t+1}
    &= \frac{1}{N}\!\left((1-\gamma)\mathbf{1}_N^\top+\gamma\mathbf{1}_N^\top\right)\mathbf{v}_t
      + \frac{\gamma a}{N}\sum_{i=1}^N (y_{i,t})^3,
\end{aligned}
\]
where the first term simplifies to $\tfrac{1}{N}\mathbf{1}_N^{\top}\mathbf{v}_t=\bar{v}_t$, and the nonlinear term $\tfrac{\gamma a}{N}\sum_{i=1}^N (y_{i,t})^3$ becomes zero whenever
the trajectory lies in the symmetric subspace $\mathcal{S}$, as established in Lemma~\ref{lem:3.1}. Consequently, the network average opinion is conserved as $\bar{v}_{t+1}=\bar{v}_t$. Hence, the opinion state admits the decomposition as
$\mathbf{v}_t=\bar{v}_t\,\mathbf{1}_N+\mathbf{z}_t$.

Applying the interaction-induced deviation operator $\mathbf{C}$ to the decomposed state $\mathbf{v}_t$ gives
$\mathbf{y}_t\!=\!\mathbf{C}\mathbf{v}_t\!=\!\bar{v}_t\,\mathbf{C}\mathbf{1}_N\!+\!\mathbf{C}\mathbf{z}_t$. Since $\mathbf{W}\mathbf{1}_N\!=\!\mathbf{1}_N$,
$\mathbf{C}\mathbf{1}_N\!=\!(\mathbf{I}\!-\!\mathbf{W})\mathbf{1}_N\!=\!\mathbf{1}_N\!-\!\mathbf{1}_N\!=\!\mathbf{0}$, so $\mathbf{y}_t\!=\!\mathbf{C}\mathbf{z}_t$. Substituting $\mathbf{C}\!=\!\mathbf{I}\!-\!\mathbf{W}$ yields $\mathbf{y}_t\!=\!(\mathbf{I}\!-\!\mathbf{W})\mathbf{z}_t$. For the complete graph, $\mathbf{W}\!=\!\tfrac{1}{N-1}(\mathbf{J}\!-\!\mathbf{I})$, and since $\mathbf{1}_N^\top\mathbf{z}_t\!=\!0$ implies $\mathbf{J}\mathbf{z}_t\!=\!\mathbf{0}$, we obtain
$\mathbf{W}\mathbf{z}_t\!=\!\tfrac{1}{N-1}(\mathbf{0}-\mathbf{z}_t)\!=\!-\!\tfrac{1}{N-1}\mathbf{z}_t$.
Hence $\mathbf{y}_t\!=\!(\mathbf{I}\!-\!\mathbf{W})\mathbf{z}_t\!=\!\mathbf{z}_t\!-\!(-\tfrac{1}{N-1}\mathbf{z}_t)\!=\!(1+\tfrac{1}{N-1})\mathbf{z}_t\!=\!\tfrac{N}{N-1}\mathbf{z}_t$.
Defining $c\coloneqq\tfrac{N}{N-1}$, we obtain $\mathbf{y}_t=c\,\mathbf{z}_t$, which shows that the disagreement vector $\mathbf{y}_t$ is a scaled transformation of $\mathbf{z}_t$ and therefore has qualitatively identical dynamics, including the same stability and bifurcation characteristics. In the decomposed state $\mathbf{v}_t\!=\!\bar{v}_t\,\mathbf{1}_N\!+\!\mathbf{z}_t$, as established above, the average consensus opinion is conserved, and thus $\bar{v}_t$ remains constant along symmetric trajectories. Consequently, the asymptotic behavior of the full state $\mathbf{v}_t$ is governed solely by the dynamics of $\mathbf{z}_t$.

\vspace{1mm}

\begin{remark}\label{r:3.1}
Since $\mathbf{y}_t$ has the same stability and bifurcation characteristics as $\mathbf{z}_t$, we restrict the subsequent analysis to the dynamics of the disagreement vector $\mathbf{y}_t$. Characterizing $\mathbf{y}_t$ is sufficient to determine the asymptotic state of the entire network.
\end{remark}

Using the dynamical sufficiency in Remark~\ref{r:3.1}, Lemma~\ref{lem:3.2} establishes the exact scalar reduction of the $N$-dimensional disagreement dynamics.

\vspace{1mm}

\begin{lemma}\label{lem:3.2}
For M-NOD~\eqref{eq:1} over a complete undirected graph
with uniform influence weights, the evolution of the coupled disagreement vector $\mathbf{y}_t$ along any trajectory in the subspace $\mathcal{S}$ defined by \eqref{eq:symmetric_subspace} reduces to $N$ identical scalar maps.
\end{lemma}

\vspace{1mm}

\begin{proof}
Applying the operator $\mathbf{C}$ to~\eqref{eq:1} and using
$\mathbf{y}_t=\mathbf{C}\mathbf{v}_t$ yields the dynamics of $\mathbf{y}_t$ as
\begin{equation}\label{eq:2}
\mathbf{y}_{t+1}
    = \mathbf{C}\!\left[(1-\gamma)\mathbf{I}+\gamma\mathbf{W}\right]\mathbf{v}_t
      + \gamma a\,\mathbf{C}\!\left(\mathbf{y}_t\right)^{\circ 3}.
\end{equation}
We first analyze the linear component
$\mathbf{C}\!\left[(1\!-\!\gamma)\mathbf{I}\!+\!\gamma\mathbf{W}\right]\mathbf{v}_t$ of
$\mathbf{y}_{t+1}$. Since $\mathbf{C}\!=\!\mathbf{I}\!-\!\mathbf{W}$ implies
$\mathbf{C}\mathbf{W}\!=\!\mathbf{W}\mathbf{C}$, the operator $\mathbf{C}$ commutes with
$\mathbf{W}$. This allows us to move $\mathbf{C}$ to the far right; substituting
$\mathbf{y}_t\!=\!\mathbf{C}\mathbf{v}_t$, the linear term becomes
$\left[(1\!-\!\gamma)\mathbf{I}\!+\!\gamma\mathbf{W}\right]\mathbf{y}_t$. As established in Lemma~\ref{lem:3.1}, the bipartite clustering and the mean-zero
property $\mathbf{1}_N^\top\mathbf{y}_t\!=\!0$ imply
$\mathbf{J}\mathbf{y}_t\!=\!\mathbf{0}$. Multiplying
$\mathbf{W}\!=\!\tfrac{1}{N-1}(\mathbf{J}\!-\!\mathbf{I})$ by $\mathbf{y}_t$ then gives
$\mathbf{W}\mathbf{y}_t\!=\!\tfrac{1}{N-1}(\mathbf{0}\!-\!\mathbf{y}_t)
\!=\!-\tfrac{1}{N-1}\mathbf{y}_t$. Substituting
$\mathbf{W}\mathbf{y}_t\!=\!-\tfrac{1}{N-1}\mathbf{y}_t$ into
$[(1\!-\!\gamma)\mathbf{I}\!+\!\gamma\mathbf{W}]\mathbf{y}_t$ and rearranging yields
$[1\!-\!\gamma\tfrac{N}{N-1}]\mathbf{y}_t$, equivalently
$\mathbf{y}_t(1-\bar{\gamma})$ with
$\bar{\gamma}\!=\!\gamma\tfrac{N}{N-1}$.

Now we analyze the nonlinear part $\gamma a\,\mathbf{C}(\mathbf{y}_t)^{\circ 3}$. From
Lemma~\ref{lem:3.1}, for any trajectory in the symmetric subspace
$\mathcal{S}$, we have $\mathbf{1}_N^\top(\mathbf{y}_t)^{\circ 3}\!=\!0$, which implies
$\mathbf{J}(\mathbf{y}_t)^{\circ 3}\!=\!\mathbf{0}$. Using $\mathbf{C}\!=\!\mathbf{I}-\mathbf{W}$ together with
$\mathbf{W}\!=\!\tfrac{1}{N-1}(\mathbf{J}\!-\!\mathbf{I})$, we first obtain
$\mathbf{C}(\mathbf{y}_t)^{\circ 3}
\!=\! (\mathbf{y}_t)^{\circ 3}
\!-\! \tfrac{1}{N-1}\!(\mathbf{J}(\mathbf{y}_t)^{\circ 3}\!-\!\mathbf{I}(\mathbf{y}_t)^{\circ 3})$.
Since $\mathbf{J}(\mathbf{y}_t)^{\circ 3}\!=\!\mathbf{0}$, this becomes
$(\mathbf{y}_t)^{\circ 3}\!
-\! \tfrac{1}{N-1}\!(\mathbf{0}\!-\!(\mathbf{y}_t)^{\circ 3})$.
This expression simplifies to
$(1\!+\!\tfrac{1}{N-1})(\mathbf{y}_t)^{\circ 3}$,
and ultimately evaluates to
$\tfrac{N}{N-1}(\mathbf{y}_t)^{\circ 3}$. With
$\bar{\gamma}\!=\!\gamma {N}/({N-1})$, the nonlinear term becomes
$\gamma a (c(\mathbf{y}_t)^{\circ 3})
\!=\! (\gamma c)a(\mathbf{y}_t)^{\circ 3}
\!=\! \bar{\gamma}\,a\,(\mathbf{y}_t)^{\circ 3}$. Combining the reduced linear and nonlinear components and substituting back into \eqref{eq:2}, we obtain
\begin{equation}\label{eq:3}
\mathbf{y}_{t+1}
    \!=\! (1 \!-\! \bar{\gamma})\,\mathbf{y}_t
      \!+\! \bar{\gamma}\,a\,(\mathbf{y}_t)^{\circ 3}.
\end{equation}
Since every agent's disagreement opinion evolves according to the same update rule, the network dynamics reduce to the scalar recursion,
\begin{equation}\label{eq:4}
y_{t+1}\!=\! (1\!-\!\bar{\gamma})\,y_t\!+\! \bar{\gamma}\,a\,y_t^{3}.
\end{equation}
Thus, the high-dimensional disagreement vector $\mathbf{y}_t$ reduces to $N$
identical scalar maps of the form \eqref{eq:4}. 
\end{proof}

Linearizing the reduced map \eqref{eq:4} at the zero-disagreement fixed point $y=0$ establishes the local consensus stability condition presented in Proposition~\ref{lem:3.3}.

\vspace{1mm}

\begin{proposition}\label{lem:3.3}
M-NOD~\eqref{eq:1} converges asymptotically to the consensus state (Definition~\ref{def:2.3}), i.e., $\lim_{t\to\infty}\mathbf{v}_t=\bar{v}_0\mathbf{1}_N$, for the reactivity rate range $0<\gamma<{2(N-1)}/{N}$.
\end{proposition}

\vspace{1mm}

\begin{proof}
By Remark~\ref{r:3.1}, convergence of the collective opinion state $\mathbf{v}_t$ to
consensus is equivalent to asymptotic stability of the disagreement vector
$\mathbf{y}_t$ at the origin ($\mathbf{y}_t\!=\!\mathbf{0}$). Consequently, we analyze
the local stability of the reduced scalar recursion $y_t$ of $\mathbf{y}_t$, given in~\eqref{eq:4} from Lemma~\ref{lem:3.2}. The Jacobian of the scalar map $f(y)\!=\!(1\!-\!\bar{\gamma})\,y\!+\!\bar{\gamma}a\,y^{3}$ from \eqref{eq:4} is obtained by taking the first derivative, $J(y)\!=\!f^{\prime}(y)\!=\!(1\!-\!\bar{\gamma})\!+\!3\bar{\gamma}a\,y^{2}$. At the consensus equilibrium $y^{\ast}=0$, this yields
$J(0)\!=\!(1\!-\!\bar{\gamma})\!+\!3\bar{\gamma}a(0)^{2}\!=\!1\!-\!\bar{\gamma}$. For local asymptotic
stability, we require $|J(0)|\!<\!1$ (Theorem 1.13 in \cite{b20}), i.e.\ $|1\!-\!\bar{\gamma}|\!<\!1$, which implies
$0\!<\!\bar{\gamma}\!<\!2$. 

Substituting $\bar{\gamma}\!=\!\gamma\tfrac{N}{N-1}$ into
$0\!<\!\bar{\gamma}\!<\!2$ gives $0<\gamma\tfrac{N}{N-1}\!<\!2$. Multiplying the inequality by
$\tfrac{N-1}{N}\!>\!0$ (since $N\!>\!0$) produces
$0\!<\!\gamma\!<\!\tfrac{2(N-1)}{N}$. This establishes the admissible reactivity-rate range for which M-NOD~\eqref{eq:1} converges asymptotically to consensus (Definition~\ref{def:2.3}). Since
$y_t\!$ goes to $0$ implies $\|\mathbf{y}_t\|\!$ approaches $0$, it follows from the decomposed state
$\mathbf{v}_t\!=\!\bar{v}_t\mathbf{1}_N\!+\!\mathbf{z}_t$ that
$\|\mathbf{v}_t\!-\!\bar{v}_0\mathbf{1}_N\|\!$ goes to $0$.
\end{proof}

Exceeding the consensus stability bound in Proposition~\ref{lem:3.3} marks the onset of instability. Theorem~\ref{thm:3.1} establishes that crossing this threshold induces a supercritical flip bifurcation where the consensus equilibrium fractures.

\vspace{1mm}

\begin{theorem}\label{thm:3.1}
At the critical reactivity rate $\gamma_c={2(N-1)}/{N}$, the consensus
equilibrium of M-NOD \eqref{eq:1} undergoes a supercritical flip
bifurcation. Consequently, for all $\gamma>\gamma_c$, this equilibrium loses stability, forcing the collective opinion state to depart from the invariant consensus average $\bar{v}_0$.
\end{theorem}

\vspace{1mm}

\begin{proof}
By Remark~\ref{r:3.1}, the qualitative dynamics of the collective opinion state are fully determined by the disagreement vector $\mathbf{y}_t$. Moreover,
Lemma~\ref{lem:3.2} shows that the evolution of $\mathbf{y}_t$ within the
symmetric subspace is governed exactly by the scalar map $y_{t+1}\!=\!f(y_t)$ from
\eqref{eq:4}. Consequently, analyzing this representative scalar recursion is
sufficient to prove that the network undergoes a supercritical flip
bifurcation at the critical scaled reactivity rate $\bar{\gamma}_c\!=\!2$ derived next. For the critical reactivity rate threshold $\gamma_c\!=\!\tfrac{2(N-1)}{N}$, we substitute into the scaled reactivity rate $\bar{\gamma}\coloneqq\gamma\!\tfrac{N}{N-1}$ from~\eqref{eq:4}. This gives $\bar{\gamma}_c\!=\!\gamma_c\left(\tfrac{N}{N-1}\right)\!=\!\left(\tfrac{2(N-1)}{N}\right)\tfrac{N}{N-1}\!=\!2$.

For the scalar map $f(y;\bar{\gamma})\!=\!(1\!-\!\bar{\gamma})y\!+\!\bar{\gamma}ay^{3}$, we
show that a \emph{flip (period-doubling) bifurcation} occurs at the fixed point
$y^{\ast}\!=\!0$ when the scaled reactivity rate reaches the threshold
$\bar{\gamma}_c\!=\!2$. This is established by verifying the four formal genericity conditions for a flip bifurcation as stated in Theorem~4.3 of~\cite{b19}.\vspace{1mm}\\
1) \textit{Invariance:} For all $\bar{\gamma}$, we have $f(0;\bar{\gamma})=0$, so the origin $y^{\ast}=0$ remains a solution for all $\bar{\gamma}$.\vspace{1mm}\\
2) \textit{Critical eigen-spectrum:} The first derivative w.r.t.\ $y$ is
    $f_y(0;\bar{\gamma})=1-\bar{\gamma}$. At $\bar{\gamma}=\bar{\gamma}_c=2$, we obtain
    $f_y(0;2)=1-2=-1$. Exactly at $\bar{\gamma}_c=2$, and hence
    $\gamma_c=\tfrac{2(N-1)}{N}$, the consensus equilibrium loses stability.\vspace{1mm}\\
3) \textit{Transversality:} Differentiating $f_y(0;\bar{\gamma})$ w.r.t.\
    $\bar{\gamma}$ yields $\tfrac{\partial}{\partial\bar{\gamma}}\!\left(1-\bar{\gamma}\right)\!=\!-\!1\!\neq\! 0$. The eigenvalue $\lambda(\bar{\gamma})\coloneqq f_y(0;\bar{\gamma})=1-\bar{\gamma}$ therefore crosses the unit-circle boundary ($\lambda=-1$) transversally with non-zero speed. This guarantees that the bifurcation is generic and not a degenerate tangency.\vspace{1mm}\\
4) \textit{Symmetry and Non-degeneracy:} The map satisfies the odd symmetry
    $f(-y;\bar{\gamma})\!=\!-\!f(y;\bar{\gamma})$, implying $f_{yy}(0;\bar{\gamma})\!=\!0$
    for all~$\bar{\gamma}$. The first non-vanishing nonlinearity is cubic: $f_{yyy}(0;\bar{\gamma})\!=\!6\bar{\gamma}a$, and hence $f_{yyy}(0;2)\!=\!12a\!\neq\! 0$. This guarantees that the quadratic term in the local normal form vanishes (no fold bifurcation) and that a cubic leading nonlinearity is present.
    
Finally, we establish the supercritical nature of the flip bifurcation by using the odd symmetry of the map to compute the cubic normal-form coefficient, in
accordance with Theorem~7.10 of~\cite{b19}. Let
$\mu\coloneqq\bar{\gamma}\!-\!2$ and rewrite the map near $(y,\mu)\!=\!(0,0)$ as
$y_{t+1}\!=\!f(y_t;2+\mu)\!=\!-\!(1\!+\!\mu)y_t\!+\!(2\!+\!\mu)a\,y_t^{3}$. By the standard
\text{1D} flip normal-form reduction, $f$ is locally topologically
conjugate to the normal form (Theorem~7.10 of~\cite{b19}) given by,
\begin{equation}\label{eq:5}
    \xi_{t+1} = -(1+\mu)\xi_t + C \xi_t^{3} + \mathcal{O}(\xi_t^{5}),
\end{equation}
where the cubic normal-form coefficient~$C$ is determined by the Schwarzian
derivative of the map at the bifurcation point. For a map with vanishing
quadratic term ($f_{yy}(0;2)\!=\!0$), this coefficient is
$C\!=\!\tfrac{1}{6}f_{yyy}(0;2)+\bigl(\tfrac{1}{2}f_{yy}(0;2)\bigr)^{2}
\!=\!\tfrac{1}{6}(12a)+0\!=\!2a$. Since $a>0$, we obtain $C=2a>0$, proving that
the flip is supercritical. 
\end{proof}

As established in Theorem~\ref{thm:3.1}, the zero-disagreement state loses stability for $\bar{\gamma}\!>\!\bar{\gamma}_{c}$. Theorem~\ref{thm:3.2} subsequently proves that this instability bifurcates into a stable, unique period-2 orbit, the discrete-time analogue of a limit cycle in continuous-time dynamical systems.

\vspace{1mm}

\begin{theorem}\label{thm:3.2}
Under M-NOD \eqref{eq:1} over an undirected complete graph with uniform influence weights, the network admits a unique symmetric period-2 orbit
$\mathcal{O}\!\coloneqq\!\{\bar{\mathbf{v}},-\bar{\mathbf{v}}\}$ satisfying
$\mathbf{v}_{t+2}\!=\!\mathbf{v}_t\!\neq\!\mathbf{v}_{t+1}$, with disagreement
amplitude $p\!\coloneqq\!\|\mathbf{C}\bar{\mathbf{v}}\|_{\infty}
\!=\!\sqrt{{(\bar{\gamma}-2})/{\bar{\gamma}\,a}}$. This period-2 orbit is locally
asymptotically stable over the reactivity-rate range:
\begin{equation}
\label{eq:polarization_gamma_range}
\frac{2(N-1)}{N}
<
\gamma
<
\frac{3(N-1)}{N}.
\end{equation}
Consequently, the network undergoes a synchronized flip in which the two opposing clusters alternate between the disagreement states $+p$ and $-p$ at every time step.
\end{theorem}

\vspace{1mm}

\begin{proof}
Consistent with Lemma~\ref{lem:3.2}, the network dynamics within the
symmetric subspace decouple into $N$ identical scalar maps $y_{t+1}=f(y_t;\bar{\gamma})$ given by~\eqref{eq:4}. Consequently, the existence, uniqueness, and stability of
a period-2 orbit for the collective opinion state $\mathbf{v}_t$ can be fully
established by proving the existence of a stable fixed point of the second iterate
scalar map $y_{t+2}=f(f(y_t;\bar{\gamma}))$.\vspace{1mm}

\noindent (a) \underline{Existence and Uniqueness:} To establish the existence of a period-2 orbit for the scalar map $y_{t+1}\!=\!f(y_t;\bar{\gamma})\!=\!(1\!-\!\bar{\gamma})y_t\!+\!\bar{\gamma}a\,y_t^{3}$, note that the map is odd as $f(-y;\bar{\gamma})\!=\!-\!f(y;\bar{\gamma})$. Consider a non-zero point $y_t\!\neq\! 0$ that satisfies the defining invariance condition $y_{t+1}\!=\!-y_t$ for the existence of a symmetric period-2 orbit. Substituting this point into the $y_{t+1}\!=\!f(y_t;\bar{\gamma})$ update~\eqref{eq:4} gives $-\!y_t\!=\!(1\!-\!\bar{\gamma})y_t\!+\!\bar{\gamma}a\,y_t^{3}$, which rearranges to $(\bar{\gamma}-2)y_t\!=\!\bar{\gamma}a\,y_t^{3}$. For $y_t\!\neq\!0$, division by $y_t$ yields the amplitude condition $y_t^{2}\!=\!\tfrac{\bar{\gamma}-2}{\bar{\gamma}a}$. In the post-bifurcation regime $\bar{\gamma}\!>\!2$ (Theorem~\ref{thm:3.1}) with $a\!>\!0$, the right-hand side is strictly positive, guaranteeing the existence of exactly one pair of real non-zero roots $y_t\!=\!\!\pm p\!=\!\!\pm\!\sqrt{\tfrac{\bar{\gamma}-2}{\bar{\gamma}a}}$. Therefore, the only nonzero solutions of
$f(y_t;\bar{\gamma})\!=\!-y_t$ are $y_t\!\!=\!\!\pm p$, so
$\mathcal{O}_y\!\!=\!\!\{p,-p\}$ is the unique symmetric period-2 orbit.

To map the scalar results back to the collective opinion state, recall the orthogonal decomposition $\mathbf{v}_t=\bar{v}_0\mathbf{1}_N+\mathbf{z}_t$. For zero invariant initial average opinion, i.e., $\bar v_t=\bar{v}_0=0$, we have $\mathbf{v}_t=\mathbf{z}_t$. Combining this with $\mathbf{z}_t=\frac{1}{c}\mathbf{y}_t$, established in the preceding analysis of Remark~\ref{r:3.1}, yields $\mathbf{v}_t=\mathbf{z}_t=\frac{1}{c}\mathbf{y}_t$. Since the scalar period-2 condition is $y_{t+1}\!=-\!y_t$, we have the corresponding vector relation $\mathbf{y}_{t+1}\!=-\!\mathbf{y}_t$. For $y_t=\pm p$, this yields
$\mathbf{v}_t\!=\!\tfrac{p}{c}\mathbf{u}$ and $\mathbf{v}_{t+1}\!=-\!\tfrac{p}{c}\mathbf{u}$, where $\mathbf{u}\!\in\!\{-1,+1\}^N$ is a balanced bipartite sign vector. Computing the second update gives $\mathbf{v}_{t+2}\!=-\!\tfrac{p}{c}(-\mathbf{u})\!=\!\tfrac{p}{c}\mathbf{u}\!=\!\mathbf{v}_t$, which satisfies the period-2 criterion $\mathbf{v}_{t+2}\!=\!\mathbf{v}_t\!\neq\!\mathbf{v}_{t+1}$ as stated in the first point of Definition~\ref{def:2.1}. Thus, the scalar orbit $\mathcal{O}_y\!=\!\{p,-p\}$ induces a vector orbit $\mathcal{O}\!=\!\{\bar{\mathbf{v}},-\bar{\mathbf{v}}\}$ with
$\bar{\mathbf{v}}\!\coloneqq\!\tfrac{p}{c}\mathbf{u}$, establishing the existence of the symmetric period-2 orbit. Since $\mathbf u$ is balanced, $\mathbf1_N^\top\mathbf u=0$, implying
$\mathbf1_N^\top\bar{\mathbf v}
=\tfrac{p}{c}\mathbf1_N^\top\mathbf u=0$.
Moreover, $p>0$ implies $\bar{\mathbf v}\neq\mathbf0$.
Hence, the orbit satisfies Condition~2 of
Definition~\ref{def:2.1}. Note that for a nonzero invariant initial average opinion, the orbit is centered at $\bar{v}_0\mathbf{1}_N$ rather than the origin. The network retains the same period-2 dynamics, with the orbit $\mathcal{O}$ shifted to $\{\bar{v}_0\mathbf{1}_N+\overline{\mathbf{z}},\,\bar{v}_0\mathbf{1}_N-\overline{\mathbf{z}}\}$.\vspace{1mm}

\noindent (b) \underline{Stability}: The local stability of the period-2 orbit is determined by the Floquet multiplier of the second-iterate map $y_{t+2}=f(f(y_t;\bar{\gamma}))$. By the Chain Rule and symmetry of the orbit ($y_{t+1}\!=\!-y_t$), the multiplier is $\mu\!=\!f'(p)\cdot f'(-p)\!=\![f'(p)]^2$. The Jacobian of the scalar map is $J\!=\!f'(y)\!=\!(1-\bar{\gamma})+3\bar{\gamma}a\,y^2$. Substituting the amplitude condition $\bar{\gamma}a\,p^2\!=\!\bar{\gamma}-2$ derived in Part (a) yields $f'(p)\!=\!(1-\bar{\gamma})+3(\bar{\gamma}-2)\!=\!2\bar{\gamma}-5$. For asymptotic stability, the magnitude of the multiplier must satisfy $|\mu|\!<\!1$, which implies $|f'(p)|\!<\!1$ (Theorem 1.13 in \cite{b20}). The inequality $-1\!<\!2\bar{\gamma}-5\!<\!1$ simplifies to $2\!<\!\bar{\gamma}\!<\!3$. Substituting $\bar{\gamma}\!=\!\gamma\frac{N}{N-1}$ gives reactivity rate bounds \eqref{eq:polarization_gamma_range} for asymptotically stable period-2 orbit network behavior.
\end{proof}

Note that at the upper bound
$\gamma\!=\!\tfrac{3(N-1)}{N}$ of
\eqref{eq:polarization_gamma_range}, the Floquet multiplier of the
period-2 orbit reaches $+1$. Beyond this threshold, the symmetric
period-2 orbit is unstable, and the resulting dynamics are beyond
the scope of this paper.

\vspace{1mm}

\begin{theorem}\label{thm:3.3}
Under M-NOD \eqref{eq:1} over an undirected complete graph with uniform influence weights, the network bifurcates as the reactivity rate $\gamma$ increases, transitioning from consensus to asymptotically stable symmetric dynamic polarization (Definition~\ref{def:2.1}).
\end{theorem}

\vspace{1mm}

\begin{proof}
The proof follows the stability transition of the system. First, invoking the scalar reduction from Lemma~\ref{lem:3.2}, Proposition~\ref{lem:3.3} shows that for $0\!<\!\gamma\!<\!\tfrac{2(N-1)}{N}$, the disagreement vector $\mathbf{y}_t$ converges to the origin. This implies that $\mathbf{v}_t$ converges to $\bar{v}_0\mathbf{1}_N$, confirming that the network reaches stable consensus. Again using Lemma~\ref{lem:3.2}, Theorem~\ref{thm:3.1} establishes that at $\gamma\!=\!\gamma_c\!=\!\tfrac{2(N-1)}{N}$, the zero-disagreement equilibrium loses stability via a supercritical flip bifurcation.

For the high-reactivity regime ($\gamma\!>\!\tfrac{2(N-1)}{N}$), Theorem~\ref{thm:3.2} guarantees the existence of a unique, stable, symmetric period-2 orbit within the bounded stability interval $\tfrac{2(N-1)}{N}\!<\!\gamma\!<\!\tfrac{3(N-1)}{N}$, which directly satisfies Condition~1 of Definition~\ref{def:2.1}. Additionally, Theorem~\ref{thm:3.2} also shows that this orbit satisfies Condition~2 of Definition~\ref{def:2.1}. Moreover, Lemma~\ref{lem:3.1} ensures that
the corresponding state lies in $\mathcal S$ and exhibits balanced
bipartite clustering. Hence, the network dynamics meet the full criteria of Definition~\ref{def:2.1}, establishing that the emergent behavior is precisely symmetric dynamic polarization.
\end{proof}

Intuitively, under high reactivity rates in the post-bifurcation regime, agents' opinions repeatedly overshoot the consensus equilibrium. Consequently, opposing clusters persistently alternate between bounded disagreement states.


\section{Main Results II: Emergence of Asymmetric Dynamic Polarization under Directed Graphs with Nonuniform Weights}\label{sec4}

In this section, we analyze the structural robustness of dynamic polarization in M-NOD under directed graph topologies with nonuniform influence weights, modeling such networks as sufficiently small symmetry-breaking perturbations of the undirected complete graph studied in Section~\ref{sec3}. To make the perturbation notation explicit, we rewrite the nominal matrices from Section~\ref{sec3} as $\mathbf W_0 \!:=\! \mathbf W \!=\! \frac{1}{N-1}(\mathbf J\!-\!\mathbf I)$ and $\mathbf C_0 \!:=\! \mathbf I\!-\!\mathbf W_0$.

The perturbed interaction matrix and the deviation operator are $\mathbf W_\epsilon=\mathbf W_0+\epsilon\mathbf{\Delta}$ and $\mathbf C_\epsilon=\mathbf I-\mathbf W_\epsilon=\mathbf C_0-\epsilon\mathbf{\Delta}$, respectively, where $\mathbf{\Delta}$ encodes directed weight heterogeneity and $|\epsilon|$ controls the magnitude of the perturbation. The perturbation is constructed to satisfy the zero row-sum condition $\mathbf{\Delta}\mathbf{1}_N=\mathbf{0}$, ensuring that $\mathbf{W}_\epsilon$ remains row-stochastic. This preserves the two-cluster synchronization pattern, if agents within a cluster share the same opinion, their subsequent opinions will remain identical within that cluster. Substituting $\mathbf W_\epsilon$ and $\mathbf C_\epsilon$ in~\eqref{eq:1} yields the perturbed one-step map $\mathbf v_{t+1}=F_\epsilon(\mathbf v_t)$, where
\begin{equation}
\label{eq:Feps_explicit_continuation}
F_\epsilon(\mathbf v)
=
\big[(1-\gamma)\mathbf I+\gamma(\mathbf W_0+\epsilon\mathbf\Delta)\big]\mathbf v
+
\gamma a\big[(\mathbf C_0-\epsilon\mathbf\Delta)\mathbf v\big]^{\circ 3}.
\end{equation}
The associated two-step map is defined as $G_\epsilon(\mathbf v)\coloneqq F_\epsilon(F_\epsilon(\mathbf v))$. Corresponding to the unique, locally asymptotically stable period-2 orbit $\mathcal O\coloneqq\{\bar{\mathbf v},-\bar{\mathbf v}\}$, 
established in Theorem~\ref{thm:3.2}, to establish the emergence of
asymmetric dynamic polarization (Definition~\ref{def:2.2}), we denote the
nearby perturbed period-2 orbit by
$\mathcal O_\epsilon\coloneqq\{\bar{\mathbf v}_\epsilon,\tilde{\mathbf v}_\epsilon\}$. Here, the two orbit states satisfy $\tilde{\mathbf v}_\epsilon=F_\epsilon(\bar{\mathbf v}_\epsilon)$ and $\bar{\mathbf v}_\epsilon=F_\epsilon(\tilde{\mathbf v}_\epsilon)$, or equivalently $G_\epsilon(\bar{\mathbf v}_\epsilon)=\bar{\mathbf v}_\epsilon$. Thus, proving the existence of $\mathcal O_\epsilon$ reduces to proving the existence of a nearby fixed point of the two-step map $G_\epsilon$, which is established in Lemma~\ref{lem:4.1} below.

This formulation allows us to evaluate whether dynamic polarization persists when exact anti-symmetry, uniform weights, and undirected communication are slightly broken. Specifically, we analytically bound the perturbation magnitude $\epsilon$ and the local error radius required to guarantee the persistence of dynamic polarization under directed graphs with nonuniform influence weights.

\vspace{1mm}

\begin{lemma}\label{lem:4.1}
For the perturbed M-NOD~\eqref{eq:Feps_explicit_continuation}, given the locally asymptotically stable period-2 orbit $\mathcal O_0\coloneqq\{\bar{\mathbf v},-\bar{\mathbf v}\}$ (Theorem~\ref{thm:3.2}), there exists $\bar{\epsilon}>0$ such that, for every $|\epsilon|<\bar{\epsilon}$, the perturbed two-step map $G_\epsilon$ admits a unique nearby fixed point $\bar{\mathbf v}_\epsilon$ satisfying $G_\epsilon(\bar{\mathbf v}_\epsilon)=\bar{\mathbf v}_\epsilon$, with $\bar{\mathbf v}_\epsilon$ approaching $\bar{\mathbf v}$ as $\epsilon$ approaches zero. 

Consequently, the perturbed system admits a nearby period-2 orbit $\mathcal O_\epsilon\coloneqq\{\bar{\mathbf v}_\epsilon,\tilde{\mathbf v}_\epsilon\}$ satisfying $F_\epsilon(\bar{\mathbf v}_\epsilon)=\tilde{\mathbf v}_\epsilon$, establishing the existence of asymmetric dynamic polarization under directed graphs with nonuniform influence weights.
\end{lemma}

\vspace{1mm}

\begin{proof}
Since $F_\epsilon$ in~\eqref{eq:Feps_explicit_continuation} is polynomial in $\mathbf v$ and smooth in $\epsilon$, the two-step map $G_\epsilon(\mathbf v)\coloneqq F_\epsilon(F_\epsilon(\mathbf v))$ is also smooth in $(\mathbf v,\epsilon)$. For the nominal symmetric case $\epsilon=0$, the perturbed map $F_\epsilon$ reduces to the nominal M-NOD~\eqref{eq:1}, which we denote by $F_0$, i.e., $\mathbf v_{t+1}=F_0(\mathbf v_t)$. By Theorem~\ref{thm:3.2}, $F_0$ admits the locally asymptotically stable period-2 orbit $\mathcal O=\{\bar{\mathbf v},-\bar{\mathbf v}\}$. Consequently, $F_0(\bar{\mathbf v})=-\bar{\mathbf v}$ and the corresponding two-step map satisfies $G_0(\bar{\mathbf v})=\bar{\mathbf v}$. Furthermore, letting $DG(\cdot)$ denote the Jacobian, local asymptotic stability of $\mathcal O$ implies that
$A_0\coloneqq DG_0(\bar{\mathbf v})$
is Schur stable, i.e., $\rho(A_0)<1$.

Because $G_\epsilon$ is smooth, in a local neighborhood of $(\bar{\mathbf v},0)$, we can express the perturbed map as $G_\epsilon(\mathbf v)=G_0(\mathbf v)+\epsilon S(\mathbf v,\epsilon)$, where $S$ is smooth and locally bounded. Although $\bar{\mathbf v}$ is a fixed point of the nominal map $G_0$, it need not remain a fixed point of $G_\epsilon$ for $\epsilon\neq0$. We therefore define the fixed-point residual $H(\mathbf v,\epsilon)\coloneqq G_\epsilon(\mathbf v)-\mathbf v$. Since $G_0(\bar{\mathbf v})=\bar{\mathbf v}$, we have $H(\bar{\mathbf v},0)=\mathbf{0}$. Moreover, differentiating $H$ with respect to $\mathbf v$ at $(\bar{\mathbf v},0)$ yields $D_{\mathbf v}H(\bar{\mathbf v},0)=DG_0(\bar{\mathbf v})-\mathbf I=A_0-\mathbf I$. Because $\rho(A_0)<1$, the eigenvalue $1$ is not in the spectrum of $A_0$; hence $A_0-\mathbf I$ is non-singular. By the Implicit Function Theorem, there exists $\bar{\epsilon}>0$ and a unique fixed point
$\bar{\mathbf v}_\epsilon$
depending smoothly on $\epsilon$ for
$|\epsilon|<\bar{\epsilon}$ such that $H(\bar{\mathbf v}_\epsilon,\epsilon)=\mathbf{0}$, with $\bar{\mathbf v}_\epsilon$ approaching $\bar{\mathbf v}$ as $\epsilon$ approaches $0$. Equivalently, $G_\epsilon(\bar{\mathbf v}_\epsilon)=\bar{\mathbf v}_\epsilon$. Finally, since $G_\epsilon(\bar{\mathbf v}_\epsilon)=\bar{\mathbf v}_\epsilon$, the fixed point $\bar{\mathbf v}_\epsilon$ generates the two-point orbit $\{\bar{\mathbf v}_\epsilon,F_\epsilon(\bar{\mathbf v}_\epsilon)\}$. By continuity, and since $F_0(\bar{\mathbf v})=-\bar{\mathbf v}\neq\bar{\mathbf v}$, these two points remain distinct for all sufficiently small $|\epsilon|$. 

Therefore, $\mathcal O_\epsilon=\{\bar{\mathbf v}_\epsilon,\tilde{\mathbf v}_\epsilon\}$,
with $\tilde{\mathbf v}_\epsilon=F_\epsilon(\bar{\mathbf v}_\epsilon)$,
constitutes a nontrivial period-2 orbit of the perturbed M-NOD map \eqref{eq:Feps_explicit_continuation}.
\end{proof}

Lemma~\ref{lem:4.1} established the existence and continuation of a nearby asymmetric period-2 orbit $\mathcal O_\epsilon$, with $\mathcal O_\epsilon$ approaching the nominal symmetric orbit $\mathcal O_0=\{\bar{\mathbf v},-\bar{\mathbf v}\}$ as $\epsilon$ approaches $0$. Theorem~\ref{thm:4.1} establishes its local asymptotic stability and derives explicit bounds guaranteeing the structural robustness of dynamic polarization under directed graphs with nonuniform influence weights.

\vspace{1mm}

\begin{theorem}\label{thm:4.1}
The period-2 orbit $\mathcal O_\epsilon\coloneqq\{\bar{\mathbf v}_\epsilon,\tilde{\mathbf v}_\epsilon\}$, whose existence is established in Lemma~\ref{lem:4.1} for the perturbed M-NOD~\eqref{eq:Feps_explicit_continuation} under directed graphs with nonuniform influence weights, is locally asymptotically stable whenever $|\epsilon|<\epsilon^\star\coloneqq\min\left\{\bar{\epsilon},\,1,\,{\lambda_{\min}(Q)}/{2c_2}\right\}$. The initial opinion state lies in its local basin of attraction, i.e., $0<\|\mathbf v_0-\bar{\mathbf v}_\epsilon\|<r_\epsilon$, where $r_\epsilon<\min\left\{1,\,{\lambda_{\min}(Q)}/{4c_3}\right\}$. Here, $Q$ is a symmetric positive definite matrix, while
$c_2>0$ and $c_3>0$ denote local perturbation and nonlinear-remainder
bounds, respectively. 

Consequently, asymmetric dynamic polarization
persists under sufficiently small directed weight perturbations.
\end{theorem}

\vspace{1mm}

\begin{proof}
By Lemma~\ref{lem:4.1}, the perturbed two-step map $G_\epsilon$ admits the fixed point $\bar{\mathbf v}_\epsilon$ associated with the orbit $\mathcal O_\epsilon$. It remains to prove the local asymptotic stability of this fixed point.

Let $A_0\coloneqq DG_0(\bar{\mathbf v})$. Since the nominal period-2 orbit $\mathcal O_0=\{\bar{\mathbf v},-\bar{\mathbf v}\}$ is locally asymptotically stable by Theorem~\ref{thm:3.2}, $A_0$ is Schur stable, i.e., $\rho(A_0)<1$. Hence, for any $Q=Q^\top>0$, there exists $P=P^\top>0$ satisfying
\begin{equation}
\label{eq:nominal_lyap_asym}
A_0^\top P A_0-P=-Q .
\end{equation}
Let $\lambda_Q\coloneqq\lambda_{\min}(Q)>0$. Define the perturbation error by $\mathbf e_t\coloneqq\mathbf v_t-\bar{\mathbf v}_\epsilon$. Since $\mathbf v_{t+2}=G_\epsilon(\mathbf v_t)$ and $G_\epsilon(\bar{\mathbf v}_\epsilon)=\bar{\mathbf v}_\epsilon$, we obtain $\mathbf e_{t+2}=G_\epsilon(\bar{\mathbf v}_\epsilon+\mathbf e_t)-G_\epsilon(\bar{\mathbf v}_\epsilon)$. A first-order Taylor expansion of $G_\epsilon$ around $\bar{\mathbf v}_\epsilon$ gives
\begin{equation}
\label{eq:taylor_perturbed_asym}
\mathbf e_{t+2}=A_\epsilon^\star\mathbf e_t+\mathbf r_t,
\end{equation}
where $A_\epsilon^\star\coloneqq DG_\epsilon(\bar{\mathbf v}_\epsilon)$ and the remainder satisfies $\|\mathbf r_t\|\!\le \!c_r\|\mathbf e_t\|^2$ for some local constant $c_r>0$.

Because $\bar{\mathbf v}_\epsilon$ depends smoothly on $\epsilon$ and approaches $\bar{\mathbf v}$ as $\epsilon$ approaches $0$, and $DG_\epsilon$ is smooth in $(\mathbf v,\epsilon)$, we have $A_\epsilon^\star=A_0+\mathcal O(\epsilon)$. Thus, defining $E_\epsilon\coloneqq A_\epsilon^\star-A_0$, there exists $c_E>0$ such that $\|E_\epsilon\|\le c_E|\epsilon|$ and $0<|\epsilon|\le1$. Using $A_\epsilon^\star=A_0+E_\epsilon$, the perturbation of the Lyapunov matrix satisfies
\begin{align*}
(A_\epsilon^\star)^\top P A_\epsilon^\star-P
&=
(A_0+E_\epsilon)^\top P(A_0+E_\epsilon)-P \nonumber\\
&\!=\!
A_0^\top \! P A_0-P
\!+\!
A_0^\top \! P E_\epsilon
\!+\!
E_\epsilon^\top \! P A_0
\!+\!
E_\epsilon^\top \! P E_\epsilon.
\end{align*}
Therefore, applying standard norm inequalities yields
\begin{align*}
\left\|
(A_\epsilon^\star)^\top P A_\epsilon^\star
-
A_0^\top P A_0
\right\|
&\le
2\|A_0\|\|P\|\|E_\epsilon\|
+
\|P\|\|E_\epsilon\|^2 \nonumber\\
&\le
\left(2\|A_0\|\|P\|c_E+\|P\|c_E^2\right)|\epsilon|,
\end{align*}
using $|\epsilon|\le1$, define $c_2\coloneqq2\|A_0\|\|P\|c_E+\|P\|c_E^2$. Then
\begin{equation}
\label{eq:c2_bound_asym}
\left\|
(A_\epsilon^\star)^\top P A_\epsilon^\star
-
A_0^\top P A_0
\right\|
\le
c_2|\epsilon| .
\end{equation}

Using \eqref{eq:nominal_lyap_asym} and \eqref{eq:c2_bound_asym}, along the error trajectory $\mathbf e_t$,
\begin{align}
\mathbf e_t^\top\!((A_\epsilon^\star)^\top \! P A_\epsilon^\star\!-\!P)\mathbf e_t
&\!=\!
-\mathbf e_t^\top \! Q\mathbf e_t
\!+\!
\mathbf e_t^\top\!\! ((A_\epsilon^\star)^\top \! P A_\epsilon^\star\!-\!A_0^\top \! P A_0)\mathbf e_t \nonumber\\
&\le
-\lambda_Q\|\mathbf e_t\|^2
+
c_2|\epsilon|\|\mathbf e_t\|^2 \nonumber\\
&=
-(\lambda_Q-c_2|\epsilon|)\|\mathbf e_t\|^2 .
\label{eq:matrix_negative_pre_asym}
\end{align}
If $|\epsilon|\!<\!\lambda_Q/(2c_2)$, then $\mathbf e_t^\top\!\left[(A_\epsilon^\star)^\top \! P A_\epsilon^\star\!-\!P\right]\mathbf e_t
\! \le \!
-\frac{\lambda_Q}{2}\|\mathbf e_t\|^2$.
\indent Since the period-2 orbit is a fixed point of $G_\epsilon$, its stability is analyzed through the two-step error dynamics $\mathbf e_t\mapsto\mathbf e_{t+2}$. Consider the quadratic Lyapunov function $V_\epsilon(\mathbf e_t)\! \coloneqq \! \mathbf e_t^\top P\mathbf e_t$. Substituting~\eqref{eq:taylor_perturbed_asym} into the Lyapunov difference gives
\begin{align*}
V_\epsilon(\mathbf e_{t+2})-V_\epsilon(\mathbf e_t)
&=
(A_\epsilon^\star\mathbf e_t+\mathbf r_t)^\top
P
(A_\epsilon^\star\mathbf e_t+\mathbf r_t)
-
\mathbf e_t^\top P\mathbf e_t \\
= \hspace{2.0em}&\!\!\!\!\!\!\!\!\!
\mathbf e_t^\top\!((A_\epsilon^\star)^\top \!P A_\epsilon^\star\!-\!P)\mathbf e_t
\!+\!2\mathbf e_t^\top \!(A_\epsilon^\star)^\top \!P\mathbf r_t
\!+\!\mathbf r_t^\top \! P\mathbf r_t.
\end{align*}
The remainder terms satisfy $
\left|
2\mathbf e_t^\top\!(A_\epsilon^\star)^\top\! P\mathbf r_t
\right|
\!\le\!
2\|A_\epsilon^\star\|\|P\|c_r\|\mathbf e_t\|^3$ and
$\left|
\mathbf r_t^\top P\mathbf r_t
\right|
\!\le\!
\|P\|c_r^2\|\mathbf e_t\|^4$. For $\|\mathbf e_t\|\!\le\!1$, $\|\mathbf e_t\|^4\!\le\!\|\mathbf e_t\|^3$. Since $A_\epsilon^\star$ remains bounded for sufficiently small $|\epsilon|$, there exists $c_3\!>\!0$ such that $V_\epsilon(\mathbf e_{t+2})\!-\!V_\epsilon(\mathbf e_t)
\!\le\!
-\frac{\lambda_Q}{2}\|\mathbf e_t\|^2
\!+\!
c_3\|\mathbf e_t\|^3$. Choose $r_\epsilon\!<\!\min\left\{1,\lambda_Q/(4c_3)\right\}$. Then, whenever $0\!<\!\|\mathbf e_t\|\!<\!r_\epsilon$, we have $c_3\|\mathbf e_t\|\!<\!\lambda_Q/4$, and the Lyapunov difference satisfies
\begin{equation}
\label{eq:strict_V_decrease_asym}
V_\epsilon(\mathbf e_{t+2})-V_\epsilon(\mathbf e_t)
\le
-\frac{\lambda_{\min}(Q)}{4}\|\mathbf e_t\|^2
\!<\!0.
\end{equation}
Thus, $\mathbf e_t=0$ is locally asymptotically stable for the two-step error dynamics (Theorem 4.1 in \cite{b21}). Since $\mathbf e_t=0$ corresponds to the fixed point $\bar{\mathbf v}_\epsilon$ of $G_\epsilon$, and $G_\epsilon$ is the two-step map of the perturbed M-NOD~\eqref{eq:Feps_explicit_continuation}, the orbit point $\bar{\mathbf v}_\epsilon$ is locally asymptotically stable under the two-step dynamics. The same argument applies to the second orbit point $\tilde{\mathbf v}_\epsilon$ by shifting the trajectory by one time step. Hence, the corresponding period-2 orbit $\mathcal O_\epsilon$ is locally asymptotically stable.
\end{proof}

In summary, the preceding analysis formally establishes that the emergence of dynamic polarization within the M-NOD framework is not a fragile artifact of perfectly symmetric network topologies. As guaranteed by Theorem~\ref{thm:4.1}, this oscillatory disagreement is structurally robust, persisting as a locally asymptotically stable asymmetric period-2 orbit even when the ideal assumptions of uniform influence weights and undirected communication are slightly perturbed.

While mathematically stable, this persistent dynamic polarization represents an undesirable state of sustained disagreement and cyclic deadlock in the collective opinion dynamics. To prevent the network from remaining trapped in such oscillatory behavior, active intervention is required. Consequently, the subsequent section introduces consensus control strategies that guarantee the elimination of both symmetric and asymmetric dynamic polarization, thereby steering the network toward a stable consensus state.


\section{Consensus Control of Dynamic Polarization}
Consider a network in a state of symmetric dynamic polarization (Definition~\ref{def:2.1}) as established in Theorem~\ref{thm:3.3}, which partitions the network into two internally synchronized opposing opinion clusters. Let these clusters be denoted by $\mathcal G_1$ and $\mathcal G_2$, with sizes $n_1=|\mathcal G_1|$ and $n_2=|\mathcal G_2|$, where $n_1+n_2=N$. Under the M-NOD dynamics~\eqref{eq:1}, the synchronized opinions of clusters $\mathcal G_1$ and $\mathcal G_2$ are denoted by $x_t$ and $z_t$, respectively.

Subsections~\ref{sec5a}, \ref{sec5b}, and~\ref{sec5c} propose three progressively less invasive opinion-control strategies to drive the network from dynamic polarization to consensus. The first strategy controls an entire polarized opinion cluster, the second controls the opinion of one agent from each polarized cluster, and the third controls the opinion of only a single agent in the entire network. This progression demonstrates that transition to consensus does not require large-scale intervention; rather, increasingly localized control actions are sufficient to break the cyclic deadlock of polarized opinion dynamics. Most notably, the final strategy establishes that influencing a single agent's opinion is mathematically sufficient to eliminate network-wide polarization and recover asymptotically stable consensus, highlighting the strong propagation of opinion influence through the network.

\vspace{-3mm}

\subsection{Network Consensus via Single-Cluster Control}\label{sec5a}
In this subsection, we propose controlling exclusively the cluster $\mathcal G_1$ by setting the reactivity rates to $\gamma_i(t)=0$ for all $i\in\mathcal G_1$, yielding $x_{t+1}=x_t=x_c$. Meanwhile, the uncontrolled cluster $\mathcal G_2$ retains the high reactivity rate $\gamma_h$ within the range~\eqref{eq:polarization_gamma_range}, under which symmetric dynamic polarization is inherently stable. Theorem~\ref{thm:5.1} establishes that controlling one polarized cluster to a fixed reference opinion suppresses the symmetric dynamic polarization, steering the uncontrolled agents to converge to this reference state and restoring a locally asymptotically stable consensus.

\vspace{1mm}

\begin{theorem}\label{thm:5.1}
For M-NOD~\eqref{eq:1} in a state of symmetric dynamic polarization (Definition~\ref{def:2.1}, Theorem~\ref{thm:3.3}), if the reactivity rates of the controlled cluster $\mathcal G_1$ are set to $\gamma_i(t)=0$ for all $i\in\mathcal G_1$, then the network is guaranteed to transition to a locally asymptotically stable consensus state (Definition~\ref{def:2.3}).
\end{theorem}

\vspace{1mm}

\begin{proof}
Proof is provided in Appendix \ref{app2}.
\end{proof}


\subsection{Consensus via Single-Agent Per-Cluster Control}\label{sec5b}
In this approach, one agent $i_1\in\mathcal G_1$ and one agent $i_2\in\mathcal G_2$, one from each polarized clusters are controlled to maintain a common reference opinion $x_c$. The remaining agents evolve according to the M-NOD~\eqref{eq:1} retain the constant reactivity rate $\gamma_h$ within the dynamic-polarization regime~\eqref{eq:polarization_gamma_range}. Under this two-cluster reduction, let $z_t$ denote the synchronized opinion shared by all uncontrolled agents.

Theorem~\ref{thm:5.2} establishes that the proposed single-agent-per-cluster control strategy guarantees a transition from dynamic polarization to full network consensus.

\vspace{1mm}

\begin{theorem}\label{thm:5.2}
For M-NOD~\eqref{eq:1} in a state of symmetric dynamic polarization (Theorem~\ref{thm:3.3}) with a total network size $N>3$, if the opinion of one agent from each polarized cluster, $\mathcal G_1$ and $\mathcal G_2$, is controlled to maintain any common reference opinion, then the network is guaranteed to transition from symmetric dynamic polarization to a locally asymptotically stable consensus (Definition~\ref{def:2.3}).
\end{theorem}

\begin{proof}
Proof is provided in Appendix \ref{app3}.
\end{proof}

\vspace{-2mm}

\begin{figure*}[t]
\centering

\includegraphics[width=0.48\textwidth]{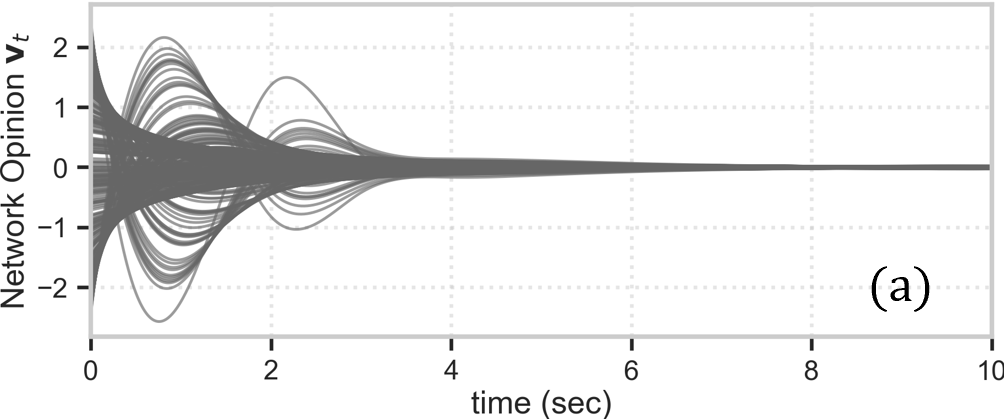}
\hfill
\includegraphics[width=0.48\textwidth]{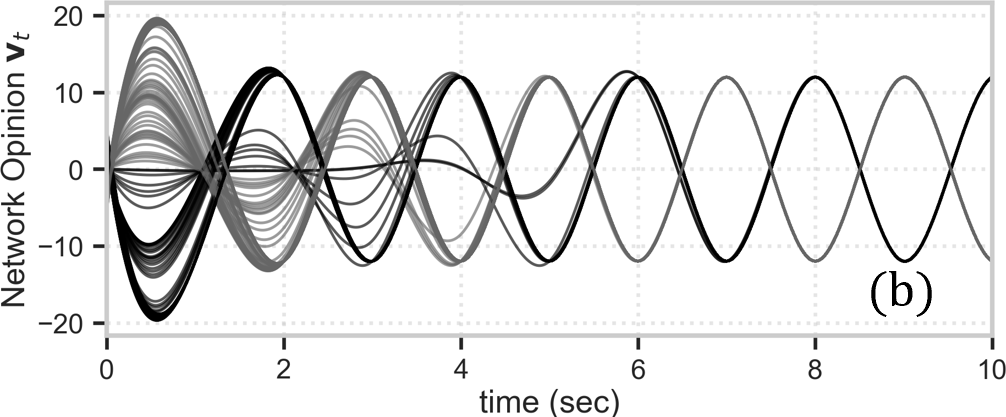}

\vspace{1mm}

\includegraphics[width=0.48\textwidth]{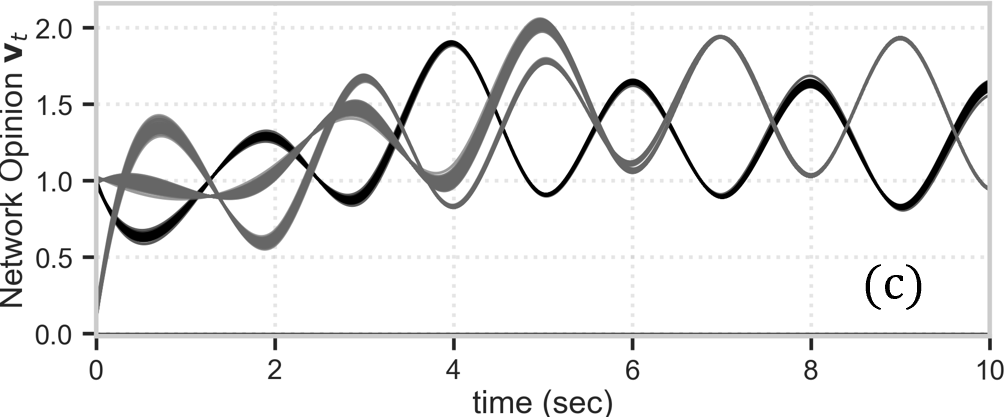}
\hfill
\includegraphics[width=0.48\textwidth]{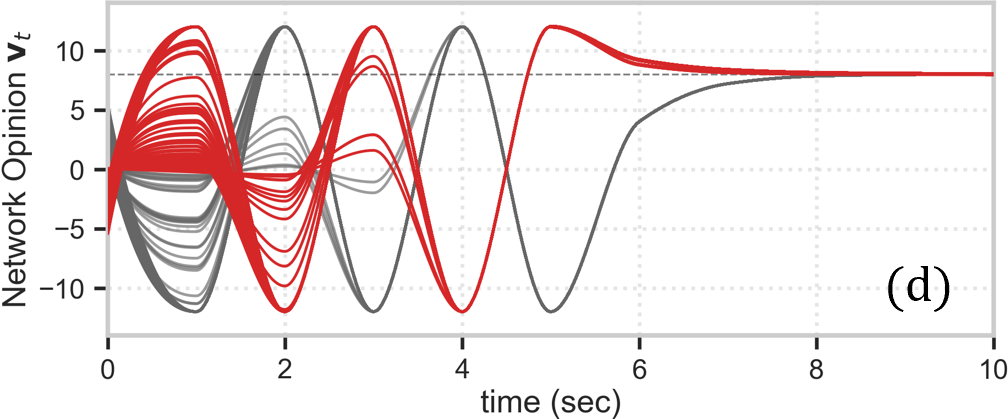}

\vspace{1mm}

\includegraphics[width=0.48\textwidth]{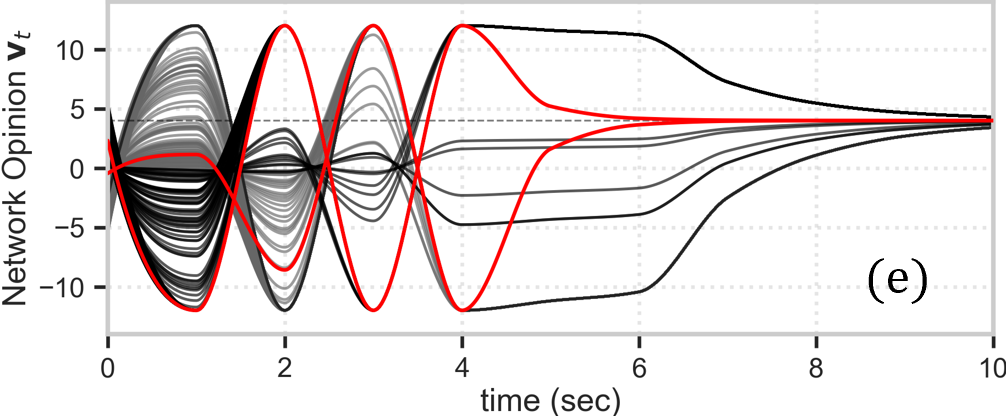}
\hfill
\includegraphics[width=0.48\textwidth]{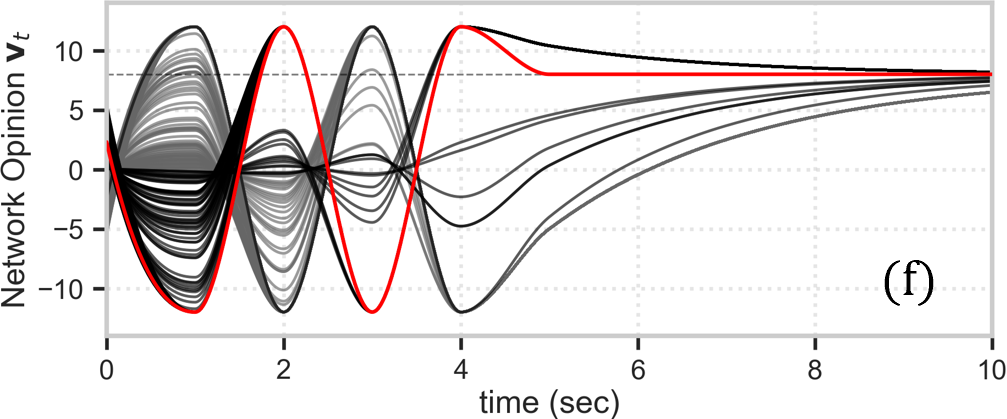}

\vspace{1mm}

\includegraphics[width=0.48\textwidth]{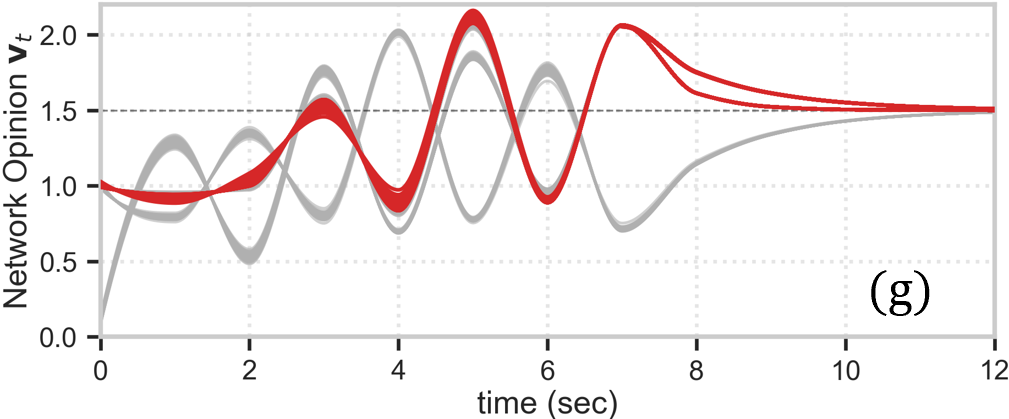}
\hfill
\includegraphics[width=0.48\textwidth]{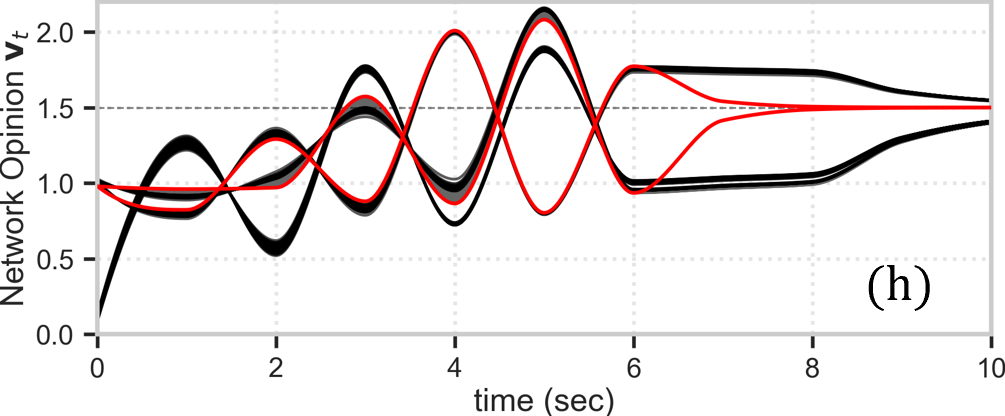}

\caption{Simulation results substantiating the theoretical emergence and consensus control of dynamic polarization in M-NOD with $N=58$, $a=0.6$, and initial opinions sampled uniformly from $[-4,4]$. (a) Consensus under the low-reactivity regime $\gamma=1.4$ (Definition~\ref{def:2.3}, Proposition~\ref{lem:3.3}). (b) Symmetric dynamic polarization exhibiting balanced bipartite clustering under $\gamma_h=2.4$ (Definition~\ref{def:2.1}, Theorem~\ref{thm:3.3}). (c) Asymmetric dynamic polarization under directed nonuniform weights with $\epsilon=0.08$ (Definition~\ref{def:2.2}, Theorem~\ref{thm:4.1}). (d)--(f) Consensus control from symmetric dynamic polarization using single-cluster control (Theorem~\ref{thm:5.1}), single-agent-per-cluster control (Theorem~\ref{thm:5.2}), and single-agent control (Theorem~\ref{thm:5.3}), respectively. (g)--(h) Consensus control from asymmetric dynamic polarization using single-cluster control and single-agent-per-cluster control, respectively (Corollary~\ref{cor:5.4}). In all control cases, the black dashed line denotes the reference opinion $x_c$, and control is activated at $t=5$ s. Together, these simulations confirm the analytically predicted bifurcation behavior and demonstrate the effectiveness of the proposed localized control strategies in breaking cyclic deadlock to restore consensus.}
\label{fig:1}
\vspace{-6mm}
\end{figure*}

\subsection{Network Consensus via Single-Agent Control}\label{sec5c}
Finally, we propose a strategy where only one agent $i_p\in\mathcal G_1$ from a single polarized cluster is controlled to maintain a fixed reference opinion $x_c$. All remaining agents evolve according to the M-NOD~\eqref{eq:1} with a constant reactivity rate $\gamma_h$ within the dynamic-polarization regime~\eqref{eq:polarization_gamma_range}. Under the corresponding reduced representation, let $z_t$ denote the synchronized opinion shared by all uncontrolled agents.

Theorem~\ref{thm:5.3} establishes that this single-agent control strategy is sufficient to guarantee a transition from dynamic polarization to full network consensus.

\vspace{1mm}

\begin{theorem}\label{thm:5.3}
For the M-NOD~\eqref{eq:1} in a state of symmetric dynamic polarization (Theorem~\ref{thm:3.3}) with a total network size $N\ge 2$, if the opinion of only one agent from one of the two polarized clusters is controlled to maintain any fixed reference opinion, then the network is guaranteed to transition from symmetric dynamic polarization to a locally asymptotically stable consensus (Definition~\ref{def:2.3}).
\end{theorem}

\vspace{1mm}

\begin{proof}
The controlled agent $i_p\in\mathcal G_1$ maintains the fixed reference opinion $x_c$, i.e., $v_{i_p}(t+1)=v_{i_p}(t)=x_c$. Let $z_t$ denote the common opinion of all uncontrolled agents $k\neq i_p$, and define the disagreement $d_t\coloneqq z_t-x_c$. Consequently, consensus control reduces to analyzing the uncontrolled-agent dynamics and proving the stability of the disagreement dynamics.

For any uncontrolled agent $k$, the deviation operator is given by $(C\mathbf v_t)_k=v_k(t)-\frac{1}{N-1}\sum_{j\neq k}v_j(t)$. Since the neighbors of agent $k$ consist of one controlled agent with opinion $x_c$ and $N-2$ uncontrolled agents with opinion $z_t$, we have $\sum_{j\neq k}v_j(t)=x_c+(N-2)z_t$. Substituting this expression into the deviation operator gives $(C\mathbf v_t)_k=z_t-\frac{1}{N-1}(x_c+(N-2)z_t)$. Rearranging terms yields $(C\mathbf v_t)_k=z_t\!\left(1-\frac{N-2}{N-1}\right)-\frac{1}{N-1}x_c$. Using $1-\frac{N-2}{N-1}=\frac{1}{N-1}$, we obtain $(C\mathbf v_t)_k=\frac{1}{N-1}(z_t-x_c)$. Defining $\alpha_p\coloneqq\frac{1}{N-1}$ and recalling that $d_t\coloneqq z_t-x_c$, it follows that $(C\mathbf v_t)_k=\alpha_p d_t$.

For every uncontrolled agent, $\gamma_k(t)=\gamma_h$. Hence, substituting $(C\mathbf v_t)_k=\alpha_p d_t$ into the M-NOD update~\eqref{eq:1} gives
\begin{equation}
\label{eq:z_update_single_agent}
z_{t+1}
=
z_t-\gamma_h\alpha_p d_t
+
\gamma_h a\alpha_p^3 d_t^3.
\end{equation}
Since $d_{t+1}=z_{t+1}-x_c$, subtracting $x_c$ from~\eqref{eq:z_update_single_agent} yields
\begin{equation}
\label{eq:d_map_single_agent}
d_{t+1}
=
(1-\gamma_h\alpha_p)d_t
+
\gamma_h a\alpha_p^3 d_t^3.
\end{equation}

Using the dynamic-polarization range~\eqref{eq:polarization_gamma_range}, the largest admissible value of $\gamma_h$ is $\bar{\gamma}_h\coloneqq\frac{3(N-1)}{N}$. Substituting $\alpha_p=\frac{1}{N-1}$ gives $\gamma_h\alpha_p<\bar{\gamma}_h\alpha_p=\frac{3(N-1)}{N}\cdot\frac{1}{N-1}=\frac{3}{N}<2$ for $N\ge 2$. Since $\gamma_h>0$ and $\alpha_p>0$, we obtain $0<\gamma_h\alpha_p<2$.

Define $\lambda_p\coloneqq1-\gamma_h\alpha_p$ and $\beta_p\coloneqq\gamma_h a\alpha_p^3$. Then $|\lambda_p|<1$, and~\eqref{eq:d_map_single_agent} becomes
$
d_{t+1}
=
\lambda_p d_t+\beta_p d_t^3.
$
Choose $\delta>0$ such that $|\lambda_p|+|\beta_p|\delta^2<1$. Then, for any $0<|d_t|<\delta$, taking absolute values in the above recursion gives $|d_{t+1}|=|\lambda_p d_t+\beta_p d_t^3|$. Using the triangle inequality yields $|d_{t+1}|\le (|\lambda_p|+|\beta_p|d_t^2)|d_t|$. Since $|d_t|<\delta$, the multiplier satisfies $|\lambda_p|+|\beta_p|d_t^2<|\lambda_p|+|\beta_p|\delta^2<1$, and therefore $|d_{t+1}|<|d_t|$.

Squaring both sides gives $d_{t+1}^2\!<\!d_t^2$. Thus, with $V_d(d_t)\coloneqq d_t^2$, we have $V_d(d_{t+1})\!<\!V_d(d_t)$ for all $0\!<\!|d_t|\!<\!\delta$. Therefore, the equilibrium $d_t\!=\!0$ is locally asymptotically stable (Theorem 4.1 in \cite{b21}) for the reduced disagreement dynamics. Since $d_t\!=\!z_t-x_c$, stability of $d_t=0$ implies $z_t\to x_c$. Because the controlled agent maintains the fixed reference opinion $x_c$, all agents converge to $x_c$, and the network locally converges to the consensus state $\mathbf v_t\!=\!x_c\mathbf 1_N$ (Definition~\ref{def:2.3}).
\end{proof}


\subsection{Consensus from Asymmetric Dynamic Polarization}
The preceding control strategies establish consensus control for the M-NOD in a state of symmetric dynamic polarization. By employing the same smooth-continuation argument used in Section~\ref{sec4}, we now establish that this consensus control persists under sufficiently small directed weight perturbations, guaranteeing suppression of asymmetric dynamic polarization and transition to consensus, established in Corollary~\ref{cor:5.4}.

\vspace{1mm}

\begin{corollary}\label{cor:5.4}
Applying the control strategies established in Theorems~\ref{thm:5.1}--\ref{thm:5.3} to the perturbed M-NOD exhibiting asymmetric dynamic polarization (Theorem~\ref{thm:4.1}) preserves the local asymptotic stability of the controlled consensus equilibrium under sufficiently small directed weight perturbations. Consequently, the network is guaranteed to transition from asymmetric dynamic polarization (Definition~\ref{def:2.2}) to consensus (Definition~\ref{def:2.3}).
\end{corollary}

\vspace{1mm}

\begin{proof}
Proof is provided in Appendix \ref{app4}.
\end{proof}


\section{Numerical Simulations}
To substantiate the theoretical results, we simulate the M-NOD~\eqref{eq:1} with $N=58$ agents over a 10 s time horizon. The nonlinear coefficient is set to $a=0.6$, and the initial opinions are sampled independently from the uniform distribution $[-4,4]$. The consensus threshold from Proposition~\ref{lem:3.3} is $\gamma_c=2(N-1)/N\approx1.966$. The stable period-2 range from Theorem~\ref{thm:3.2} is given by~\eqref{eq:polarization_gamma_range}, yielding $1.966<\gamma<2.948$. Accordingly, we choose $\gamma=1.4$ for the consensus-stable case and $\gamma_h=2.4$ for the dynamic-polarization and control cases. To avoid visual clutter, the discrete-time trajectories are rendered using univariate spline interpolation with minimal smoothing.

Fig. \ref{fig:1}(a) illustrates the network behavior within the consensus-stable reactivity regime ($\gamma=1.4$). The distinct initial opinion trajectories rapidly collapse to a single steady-state value, confirming convergence to consensus (Definition~\ref{def:2.3}). Conversely, Fig.~\ref{fig:1}(b) shows the network behavior for $\gamma_h=2.4$, which exceeds the bifurcation threshold $\gamma_c$ established in Theorem~\ref{thm:3.1}. Rather than converging to consensus, the trajectories split into two balanced clusters that persistently oscillate. This confirms the emergence of symmetric dynamic polarization and balanced bipartite clustering (Theorem~\ref{thm:3.3} and Lemma~\ref{lem:3.1}).

For asymmetric dynamic polarization, we use the perturbed M-NOD formulation~\eqref{eq:Feps_explicit_continuation}, representing a small symmetry-breaking perturbation of the undirected complete graph. The nominal interaction matrix is $\mathbf W_0$, and directed weight heterogeneity is introduced through $\mathbf W_\epsilon=\mathbf W_0+\epsilon\mathbf\Delta$, with $\epsilon=0.08$. The perturbed matrix is row-normalized after construction to preserve $\mathbf W_\epsilon\mathbf 1_N=\mathbf 1_N$. The agents are divided into $|\mathcal L|=15$ leaders and $|\mathcal F|=43$ followers. The perturbation matrix $\mathbf\Delta$ creates a directed graph with nonuniform influence weights by assigning different weights to different leader--follower interaction directions. Specifically, $\Delta_{ij}=10/|\mathcal L|$ for $i\in\mathcal F$ and $j\in\mathcal L$, $\Delta_{ij}=2/(|\mathcal L|-1)$ for $i\in\mathcal L$, $j\in\mathcal L$, and $j\neq i$, and $\Delta_{ij}=-1/|\mathcal F|$ for $i\in\mathcal L$ and $j\in\mathcal F$, with all other entries set to zero. As shown in Fig.~\ref{fig:1}(c), this structural perturbation successfully breaks the symmetry, yielding stable period-2 oscillations with unequal cluster amplitudes and a shifted mean.

\begin{figure}[t]
    \centering
    \includegraphics[width=\columnwidth]{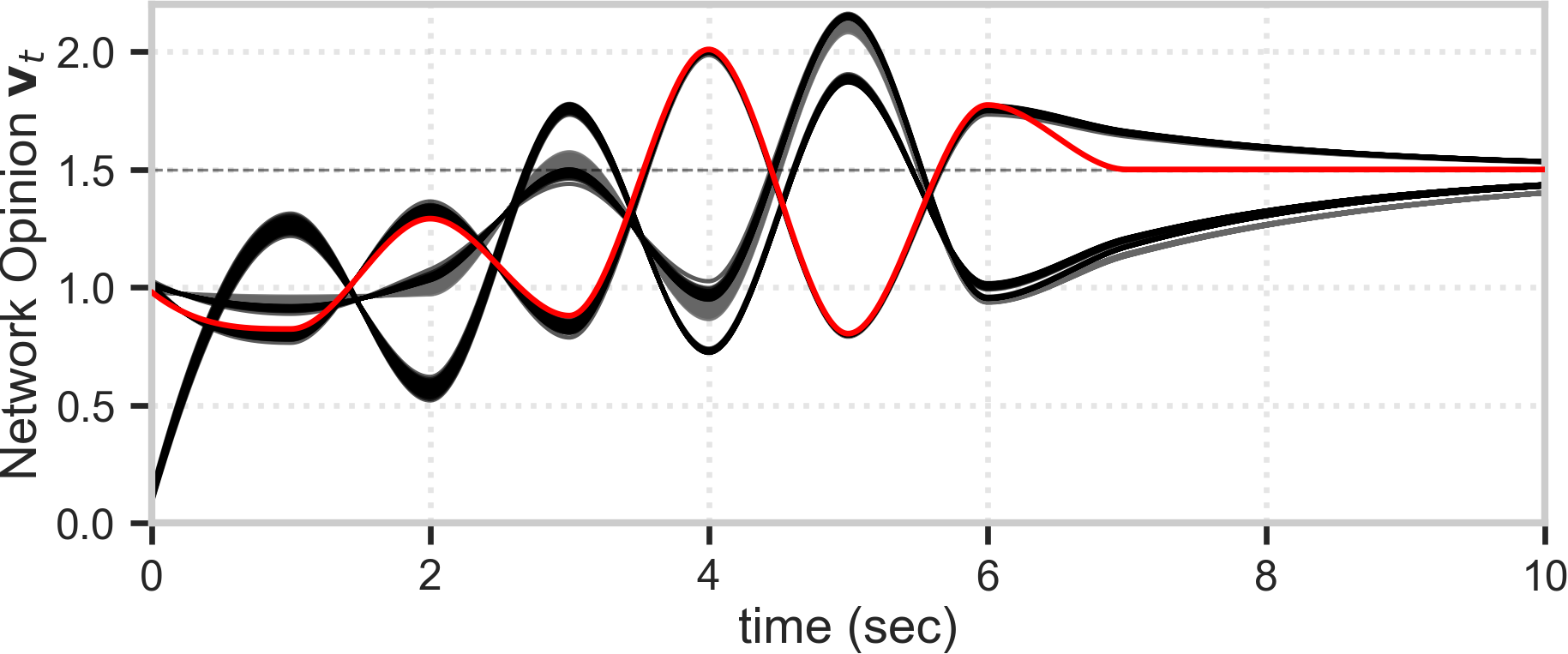}
    \vspace{-6mm}
    \caption{Simulation results validating single-agent control of asymmetric dynamic polarization under directed nonuniform weights with $\epsilon=0.08$ (Corollary~\ref{cor:5.4}, extending Theorem~\ref{thm:5.3}). The black dashed line denotes the reference opinion $x_c$, and control is activated at $t=5$ s. This confirms that anchoring the opinion of just one agent is sufficient to break asymmetric cyclic deadlock and restore network-wide consensus.}
    \label{fig:2}
    \vspace{-6mm}
\end{figure}

For the control simulations, the uncontrolled agents retain the high reactivity rate $\gamma_h=2.4$, and the target reference opinion $x_c$ is indicated by the black dashed line. Starting at $t=5$ s, after symmetric dynamic polarization is established, the respective controls are activated. Fig.~\ref{fig:1}(d) demonstrates single-cluster control (Theorem~\ref{thm:5.1}), Fig.~\ref{fig:1}(e) illustrates single-agent-per-cluster control (Theorem~\ref{thm:5.2}), and Fig.~\ref{fig:1}(f) demonstrates the single-agent control (Theorem~\ref{thm:5.3}). In all three scenarios, the interventions successfully break the cyclic deadlock, dampening the oscillations and driving the network toward the reference opinion $x_c$. Analogous results for the control of asymmetric dynamic polarization are shown in Fig. \ref{fig:1}(g), Fig. \ref{fig:1}(h), and Fig. \ref{fig:2}, confirming that all three localized control strategies remain effective under directed nonuniform interaction weights. In each case, the control input eliminates persistent oscillatory disagreement and restores asymptotically stable consensus.


\section{Conclusion}
This paper proposed the Minimally-Nonlinear Opinion Dynamics (M-NOD) framework to characterize dynamic polarization as an emergent outcome of local, agent-level opinion updates. We rigorously proved that, under a mean-field baseline, the disagreement dynamics admit an exact scalar reduction. Under a high reactivity rate, the consensus equilibrium loses stability, and a supercritical flip bifurcation gives rise to a unique locally asymptotically stable periodic orbit, thereby characterizing symmetric dynamic polarization through persistent, balanced bipartite opinion clusters. Furthermore, we established structural robustness by framing directed topologies and nonuniform influence weights as symmetry-breaking perturbations, proving the existence and stability of asymmetric dynamic polarization, where opposing clusters of unequal size remain locked in an oscillatory deadlock.

We developed localized agent-level consensus control laws that progressively reduce the required intervention from controlling one polarized cluster, to one agent from each cluster, and finally to a single controlled agent. We proved that each strategy guarantees the transition to asymptotically stable consensus, culminating in the result that anchoring the opinion of only a single agent is sufficient to eliminate persistent oscillatory disagreement and restore network-wide consensus. 
Collectively, these results establish the M-NOD framework as a principled foundation for understanding the emergence of dynamic polarization and for designing minimally invasive control strategies that restore consensus in nonlinear networked systems.

\section{APPENDIX}

\subsection{Single Agent to Multi-Agent M-NOD}\label{app1}
The nonlinear opinion update rule for single agent $i$ with prediction error 
$\delta_{i,t}\!\coloneqq\!z_{i,t}-v_{i,t}$ is:
\begin{equation}\label{eq:scalar_RW}
    v_{i,t+1}
    = v_{i,t}
      + \gamma\,\delta_{i,t}
      - \gamma a\,(\delta_{i,t})^{3}.
\end{equation}
For $N$ agents learning with opinion vector $\mathbf{v}_t\!\in\!\mathbb{R}^N$, the learning target of agent~$i$ is the weighted aggregate of its neighbors,
$z_{i,t}=\sum_{j=1}^N w_{ij}v_{j,t}$. Each agent adjusts its opinion to achieve local consensus with its neighbors' opinions. Collecting these terms into a vector yields the social target $\mathbf{z}_t\!=\!\mathbf{W}\mathbf{v}_t$, where $\mathbf{W}\!=\![w_{ij}]\!\in\!\mathbb{R}_{\ge 0}^{N\times N}$ is the weighted adjacency matrix. 

Recalling the interaction-induced deviation operator $\mathbf{C}\coloneqq\mathbf{I}-\mathbf{W}$, the error vector becomes $\boldsymbol{\delta}_t = \mathbf{W}\mathbf{v}_t-\mathbf{v}_t = -\mathbf{C}\mathbf{v}_t$. Substituting this into the linear component of~\eqref{eq:scalar_RW} yields $\mathbf{v}_t+\gamma(-\mathbf{C}\mathbf{v}_t)=(\mathbf{I}-\gamma\mathbf{C})\mathbf{v}_t$. Expanding $\mathbf{C}$ allows this to be rewritten as $((1-\gamma)\mathbf{I}+\gamma\mathbf{W})\mathbf{v}_t$. Finally, evaluating the nonlinear term using $\boldsymbol{\delta}_t=-\mathbf{C}\mathbf{v}_t$ and the odd symmetry of the cubic function yields $-\gamma a(-\mathbf{C}\mathbf{v}_t)^{\circ 3}=\gamma a(\mathbf{C}\mathbf{v}_t)^{\circ 3}$. Combining these components yields the M-NOD~\eqref{eq:1}.


\subsection{Proof of Theorem \ref{thm:5.1}}\label{app2}
For every controlled agent $i\in\mathcal G_1$, the choice $\gamma_i(t)=0$ halts the opinion update, so the controlled cluster opinion satisfies $x_t=x_c$ for all times after control activation, whereas the uncontrolled cluster $\mathcal G_2$ has opinion $z_t$. Therefore, consensus control reduces to analyzing the uncontrolled cluster $\mathcal G_2$ and proving stability of the inter-cluster disagreement dynamics.

Define the inter-cluster disagreement as $d_t\coloneqq z_t-x_c$. For any uncontrolled agent $k\in\mathcal G_2$, the deviation operator is obtained by summing over all neighboring agents $j\neq k$ as
\begin{equation}
\label{eq:Cv_complete_graph_single_cluster}
(C\mathbf v_t)_k
=
v_k(t)-\frac{1}{N-1}\sum_{j\neq k}v_j(t).
\end{equation}
Since agents in $\mathcal G_1$ have opinion $x_c$ and agents in $\mathcal G_2$ have opinion $z_t$, the neighbor sum for any $k\in\mathcal G_2$ is $\sum_{j\neq k}v_j(t) = n_1x_c+(n_2-1)z_t$. Substituting this into~\eqref{eq:Cv_complete_graph_single_cluster} yields
\begin{align}
(C\mathbf v_t)_k
&=
z_t
-
\frac{1}{N-1}
[
n_1x_c+(n_2-1)z_t
] \nonumber\\
&=
z_t\left(1-\frac{n_2-1}{N-1}\right)
-
\frac{n_1}{N-1}x_c.
\label{eq:Cv_before_cluster_simplification}
\end{align}
Using $n_1+n_2=N$, we have $1-\frac{n_2-1}{N-1}=\frac{n_1}{N-1}$. Therefore,
\begin{equation}
\label{eq:Cv_alpha_d_single_cluster}
(C\mathbf v_t)_k
=
\frac{n_1}{N-1}(z_t-x_c)
=
\alpha d_t,
\;
\alpha\coloneqq\frac{n_1}{N-1}.
\end{equation}

For the uncontrolled cluster, $\gamma_k(t)=\gamma_h$ for all $k\in\mathcal G_2$. Hence, using~\eqref{eq:Cv_alpha_d_single_cluster} in the M-NOD~\eqref{eq:1} update gives
\begin{align}
z_{t+1}
&=
z_t-\gamma_h(C\mathbf v_t)_k
+
\gamma_h a(C\mathbf v_t)_k^3 \nonumber\\
&=
z_t-\gamma_h\alpha d_t
+
\gamma_h a\alpha^3 d_t^3 .
\label{eq:z_update_single_cluster}
\end{align}
Since $d_{t+1}=z_{t+1}-x_c$, subtracting $x_c$ from~\eqref{eq:z_update_single_cluster} yields the disagreement dynamics:
\begin{equation}
\label{eq:d_reduced_map_single_cluster}
d_{t+1}
=
\left(1-\gamma_h\alpha \right)d_t
+
\gamma_h a\alpha^3 d_t^3 .
\end{equation}

Under symmetric dynamic polarization, the balanced clusters satisfy $n_1\!  =\!  n_2\!  =\!  N/2$ (Lemma~\ref{lem:3.1}, Definition~\ref{def:2.1}). Substituting this into $\alpha=\frac{n_1}{N-1}$ from \eqref{eq:Cv_alpha_d_single_cluster} yields $\alpha=\frac{N}{2(N-1)}$. Using the dynamic-polarization range~\eqref{eq:polarization_gamma_range} from Theorem~\ref{thm:3.3}, the largest admissible value of $\gamma_h$ is bounded by $\bar{\gamma}_h\coloneqq\frac{3(N-1)}{N}$. Therefore, $0<\gamma_h\alpha
<\frac{3(N-1)}{N}\frac{N}{2(N-1)}=\frac{3}{2}<2$. Define $\lambda\coloneqq1-\gamma_h\alpha$ and
$\beta\coloneqq\gamma_h a\alpha^3$. It follows that
$|\lambda|<1$. Choose $\delta>0$ such that
$|\lambda|+|\beta|\delta^2<1$. Then, using
\eqref{eq:d_reduced_map_single_cluster}, for every
$0<|d_t|<\delta$,
$
|d_{t+1}|
=
|\lambda d_t+\beta d_t^3|
\le
\left(|\lambda|+|\beta|d_t^2\right)|d_t|
<
|d_t|$. Define the Lyapunov function $V_d(d_t)\coloneqq d_t^2$. Then
$V_d(d_{t+1})<V_d(d_t)$ for all $0<|d_t|<\delta$.
Therefore, $d_t=0$ is locally asymptotically stable for the
inter-cluster disagreement dynamics. Since $d_t=z_t-x_c$, stability of $d_t=0$ implies $z_t\to x_c$. Because the controlled cluster is fixed at $x_t=x_c$, both clusters converge to $x_c$, and under the two-cluster reduction, the network locally converges to the consensus state $\mathbf v_t=x_c\mathbf 1_N$ (Definition~\ref{def:2.3}). \hfill$\blacksquare$


\subsection{Proof of Theorem \ref{thm:5.2}}\label{app3}
The controlled agents $i_1\in\mathcal G_1$ and $i_2\in\mathcal G_2$ maintain the common reference opinion $x_c$, i.e., $v_{i_1}(t+1)=v_{i_1}(t)=x_c$ and $v_{i_2}(t+1)=v_{i_2}(t)=x_c$. Let $z_t$ denote the common opinion of all uncontrolled agents $k\notin\{i_1,i_2\}$, and define the disagreement $d_t\coloneqq z_t-x_c$. Consequently, consensus control reduces to analyzing the uncontrolled-agent dynamics and proving the stability of the disagreement dynamics. 

For any uncontrolled agent $k$, the deviation operator is given by $ (C\mathbf v_t)_k = v_k(t)-\frac{1}{N-1}\sum_{j\neq k}v_j(t) $. Since the complete graph gives agent $k$ exactly $N-1$ neighbors,
excluding agent $k$ itself leaves the two controlled agents and
$N-3$ other uncontrolled agents. Hence,
$\sum_{j\neq k}v_j(t)=2x_c+(N-3)z_t$. Substituting this expression into the deviation operator gives $(C\mathbf v_t)_k=z_t-\frac{1}{N-1}(2x_c+(N-3)z_t)$. Rearranging terms yields $(C\mathbf v_t)_k=z_t\!\left(1-\frac{N-3}{N-1}\right)-\frac{2}{N-1}x_c$. Using $1-\frac{N-3}{N-1}=\frac{2}{N-1}$, we obtain $(C\mathbf v_t)_k=\frac{2}{N-1}(z_t-x_c)$. Defining $\alpha_p\coloneqq\frac{2}{N-1}$ and recalling that $d_t\coloneqq z_t-x_c$, it follows that $(C\mathbf v_t)_k=\alpha_p d_t$. For every uncontrolled agent, $\gamma_k(t)=\gamma_h$. Hence, substituting $(C\mathbf v_t)_k=\alpha_p d_t$ into the M-NOD update~\eqref{eq:1} gives
\begin{equation}
\label{eq:z_update_per_cluster}
z_{t+1}
=
z_t-\gamma_h\alpha_p d_t
+
\gamma_h a\alpha_p^3 d_t^3.
\end{equation}
Since $d_{t+1}=z_{t+1}-x_c$, subtracting $x_c$ from~\eqref{eq:z_update_per_cluster} yields
\begin{equation}
\label{eq:d_map_per_cluster}
d_{t+1}
=
(1-\gamma_h\alpha_p)d_t
+
\gamma_h a\alpha_p^3 d_t^3.
\end{equation}

Using the dynamic-polarization range~\eqref{eq:polarization_gamma_range}, the largest admissible value of $\gamma_h$ is $\bar{\gamma}_h\coloneqq\frac{3(N-1)}{N}$. Substituting $\alpha_p=\frac{2}{N-1}$ gives $\gamma_h\alpha_p<\bar{\gamma}_h\alpha_p=\frac{3(N-1)}{N}\cdot\frac{2}{N-1}=\frac{6}{N}<2$ for $N>3$. Since $\gamma_h>0$ and $\alpha_p>0$, we obtain $0<\gamma_h\alpha_p<2$.

Define $\lambda_p\coloneqq1-\gamma_h\alpha_p$ and $\beta_p\coloneqq\gamma_h a\alpha_p^3$. Then $|\lambda_p|<1$, and~\eqref{eq:d_map_per_cluster} becomes
$
d_{t+1}
=
\lambda_p d_t+\beta_p d_t^3.
$
Choose $\delta>0$ such that $|\lambda_p|+|\beta_p|\delta^2<1$. Then, for any $0<|d_t|<\delta$, taking absolute values in the above recursion gives
$|d_{t+1}|\!=\!|\lambda_p d_t+\beta_p d_t^3|$. Using the triangle inequality yields $|d_{t+1}|\le (|\lambda_p|+|\beta_p|d_t^2)|d_t|$. Since $|d_t|<\delta$, the multiplier satisfies $|\lambda_p|+|\beta_p|d_t^2<|\lambda_p|+|\beta_p|\delta^2<1$, and therefore $|d_{t+1}|<|d_t|$.

Squaring both sides gives $d_{t+1}^2<d_t^2$. Thus, with $V_d(d_t)\coloneqq d_t^2$, we have $V_d(d_{t+1})<V_d(d_t)$ for all $0<|d_t|<\delta$. Therefore, the equilibrium $d_t=0$ is locally asymptotically stable (Theorem 4.1 in \cite{b21}) for the reduced disagreement dynamics. Since $d_t=z_t-x_c$, stability of $d_t=0$ implies $z_t\to x_c$. As the controlled agents maintain the reference opinion $x_c$, all agents converge to $x_c$, and the network locally converges to the consensus state $\mathbf v_t=x_c\mathbf 1_N$ (Definition~\ref{def:2.3}). \hfill$\blacksquare$


\subsection{Proof of Corollary \ref{cor:5.4}}\label{app4}
Consider the perturbed interaction matrix $\mathbf W_\epsilon=\mathbf W_0+\epsilon\mathbf\Delta$ established in Section~\ref{sec4}. The perturbation is row-normalized so that $\mathbf W_\epsilon\mathbf 1_N=\mathbf 1_N$, ensuring that the consensus manifold remains invariant. The perturbation also preserves the equitable partition; hence, the uncontrolled agents remain synchronized, and the reduced scalar disagreement coordinate $d_t$ remains valid. For each control strategy evaluated in Theorems~\ref{thm:5.1}--\ref{thm:5.3}, the nominal reduced disagreement dynamics under the complete graph $\mathbf W_0$ have the local form $d_{t+1}=\lambda_0 d_t+\mathcal O(d_t^3)$ with $|\lambda_0|<1$. Since $\mathbf W_\epsilon$ depends smoothly on $\epsilon$, the perturbed reduced dynamics can be written as $d_{t+1}=\lambda_\epsilon d_t+R_\epsilon(d_t)$, where $\lambda_\epsilon$ depends continuously on $\epsilon$ and $R_\epsilon(d_t)=\mathcal O(d_t^3)$.

By continuity, there exists $\epsilon^\star>0$ such that $|\lambda_\epsilon|<1$ for all $|\epsilon|<\epsilon^\star$. Moreover, for sufficiently small $d_t$, there exists $c_\epsilon>0$ such that $|R_\epsilon(d_t)|\le c_\epsilon |d_t|^3$. Hence, $|d_{t+1}|\!\le\! (|\lambda_\epsilon|\!+\!c_\epsilon d_t^2)|d_t|$. Choosing $\delta\!>\!0$ such that $|\lambda_\epsilon|\!+\!c_\epsilon\delta^2\!<\!1$ gives $|d_{t+1}|\!<\!|d_t|$ for all $0\!<\!|d_t|\!<\!\delta$. 
With the Lyapunov function $V(d_t)\coloneqq d_t^2$, we have $V(d_{t+1})\!<\!V(d_t)$. Hence, the consensus equilibrium is locally asymptotically stable (Theorem 4.1 in \cite{b21}). \hfill$\blacksquare$


\end{document}